\newcommand{\iffextendedver}[2]{#1}
\newcommand{\ifextendedver}[1]{\iffextendedver{#1}{}}

\newcommand{\iffwithappendix}[2]{\iffextendedver{#1}{#2}}
\newcommand{\ifwithappendix}[1]{\iffextendedver{#1}{}}

\documentclass[year=22,pdfa,conference]{\iffextendedver{IEEEtran}{fmcad}}
\usepackage[english]{babel}

\usepackage{mathptmx}

\usepackage[ruled,vlined,linesnumbered,titlenotnumbered,noend]{algorithm2e}
\usepackage{algpseudocode}

\SetAlFnt{\small}
\SetAlCapFnt{\small}
\SetAlCapNameFnt{\small}
\SetKwInput{Initial}{Initial call}
\SetKwProg{Fn}{}{}{}
\SetAlgoNoLine
\SetInd{0.6em}{0.6em}
\SetNoFillComment
\SetCommentSty{mycommfont}
\DontPrintSemicolon
\SetNlSty{textsf}{}{}

\usepackage{dirtree}

\usepackage{float}
\usepackage[font=footnotesize]{caption}
\usepackage[utf8]{inputenc}
\usepackage[T1]{fontenc}
\usepackage[scaled=0.82]{beramono}
\usepackage{xcolor}
\usepackage[nointegrals]{wasysym}
\usepackage[nocheckfootnote]{flushend}
\ifextendedver{\usepackage{orcidlink}}
\usepackage{hyperref}
\usepackage{amsmath,amssymb,amscd,stmaryrd,amsthm}
\usepackage{centernot}
\usepackage{breakurl}
\iffextendedver{
\usepackage[spacing=true,protrusion=true,expansion=false]{microtype}
}{
\usepackage[activate={true,nocompatibility},kerning=true,spacing=true]{microtype}
}
\usepackage{multirow}
\usepackage{textcomp}
\usepackage{epsfig}
\usepackage{subfig}
\usepackage{wrapfig}
\usepackage{fancybox}
\usepackage{xspace}
\usepackage{alltt}
\usepackage{paralist}
\usepackage{todonotes}
\usepackage[shortlabels]{enumitem}
\usepackage{booktabs}
\usepackage[final]{listings}
\usepackage{numprint}
\usepackage[numbers]{natbib}

\usepackage{etex}
\usepackage{soul}
\usepackage{pgfplots}
\pgfplotsset{compat=1.14}
\usepackage{pgfplotstable}
\usepackage{makecell}
\usepackage{tikz}
\usetikzlibrary{bbox}
\usetikzlibrary{backgrounds,snakes}
\usetikzlibrary{shapes,arrows,calc}
\usetikzlibrary{decorations.markings,decorations.pathreplacing}
\usetikzlibrary{shadows,positioning,fit}

\newtheorem{claim}{Claim}{\itshape}{\rmfamily}

\newtheorem{theorem}{Theorem}{\itshape}{\rmfamily}
\newtheorem{lemma}{Lemma}{\itshape}{\rmfamily}
\newtheorem{definition}{Definition}{\itshape}{\rmfamily}

\usepackage[capitalize]{cleveref}
\crefname{section}{Sect.}{Sect.}
\Crefname{section}{Section}{Sections}
\crefname{line}{line}{lines}
\crefname{challenge}{challenge}{challenges}
\crefname{invariant}{invariant}{invariants}
\Crefname{line}{Line}{Lines}
\crefname{claim}{claim}{claims}
\Crefname{claim}{Claim}{Claims}

\newlist{propenum}{enumerate}{2} \setlist[propenum]{ref=\arabic*, label={\arabic*.}}
\crefname{propenumi}{property}{properties}
\crefname{propenumii}{property}{properties}

\newlist{ruleenum}{enumerate}{2} \setlist[ruleenum]{ref=R\arabic*, label={R\arabic*.}}
\crefname{ruleenumi}{rule}{rules}
\Crefname{ruleenumi}{Rule}{Rules}

\newlist{phaseenum}{enumerate}{2} \setlist[phaseenum]{ref=P-\Roman*, label={P-\Roman*.}}
\crefname{phaseenumi}{phase}{phases}
\Crefname{phaseenumi}{Phase}{Phases}

\newcounter{desccount}

\newcommand{\descref}[1]{\hyperref[#1]{#1}}

\hypersetup{
    colorlinks,
    linkcolor={red!50!black},
    citecolor={green!50!black},
    urlcolor={blue!80!black}
}

\makeatletter
\newcommand{\removelatexerror}{\let\@latex@error\@gobble}
\makeatother

\newcommand{\tikzwrapfigbg}{\begin{pgfonlayer}{background}
    \path[fill=gray!10,rounded corners]
    (current bounding box.south west) rectangle
    (current bounding box.north east);
\end{pgfonlayer}}

\definecolor{mygreen}{rgb}{0.05, 0.5, 0.06}
\definecolor{myorange}{rgb}{0.93, 0.49, 0.1}
\definecolor{myred}{rgb}{0.82, 0.1, 0.26}
\definecolor{myblue}{rgb}{0.01, 0.28, 1.0}
\definecolor{myviolet}{rgb}{0.6, 0.4, 0.8}
\definecolor{mygray}{rgb}{0.9, 0.89, 0.89}
\definecolor{mypurple}{rgb}{0.41, 0.16, 0.38}

\newcommand\prog{{\mathbb P}}

\newcommand\conf\gamma

\newcommand\areg{\ensuremath{a}}
\newcommand\breg{\ensuremath{b}}
\newcommand\creg{\ensuremath{c}}
\newcommand\dreg{\ensuremath{d}}

\newcommand\thstate{\sigma}
\newcommand\run\rho
\newcommand\thrun\pi
\newcommand\pth\pi

\newcommand\transset{{\mathbb E}}
\newcommand\trans{e}

\newcommand\trace\tau
\newcommand\lockedtrace\sigma
\newcommand{\totorder}[1]{<_{#1}}

\newcommand\otrace\sigma

\newcommand\tlub\sqcup
\newcommand\tglb\sqcap
\newcommand\tequiv\sim
\newcommand\ttequiv\equiv
\newcommand\ctordering\sqSubset
\newcommand\tordering\sqsubseteq
\newcommand\stordering\sqsubset
\newcommand\eventset\transset

\newcommand\event\trans

\newcommand\extend\oplus

\newcommand\lbl\ell
\newcommand\silentlbl\varepsilon

\newcommand\assigned\leftarrow

\newcommand{\gbalg}{\mbox{\textsc{Optimal-DPOR-Await}}\xspace}
\newcommand{\algname}{\gbalg}

\newcommand{\keyword}[1]{\mbox{\bf #1}}

\newcommand\alphabet\Sigma
\newcommand{\tuple}[1]{\left\langle#1\right\rangle}
\newcommand{\set}[1]{\left\{#1\right\}}

\newcommand\app\bullet

\newcommand{\emptyseq}{\langle \rangle}
\newcommand\emptyword\epsilon

\newcommand\xvar{\ensuremath{\mathtt{x}}}
\newcommand\yvar{\ensuremath{\mathtt{y}}}
\newcommand\zvar{\ensuremath{\mathtt{z}}}

\newcommand\avar\areg
\newcommand\bvar\breg
\newcommand\cvar\creg
\newcommand\dvar\dreg

\newcommand\stmt[1]{\mbox{\texttt{#1}}}
\newcommand\await{\stmt{\textbf{await}}\xspace}
\newcommand\assume{\stmt{\textbf{assume}}\xspace}

\newcommand\nmodels\nvDash
\newcommand\restrict[2]{#1\raise-.5ex\hbox{\ensuremath|}_{#2}}

\newcommand\add\odot

\newcommand\bench[1]{\mbox{\small\textsf{#1}}\xspace}
\newcommand\toolnamefont[1]{\textsc{#1}\xspace}
\newcommand\CDSChecker{\toolnamefont{CDSChecker}}

\newcommand\RCMC{\toolnamefont{RCMC}}
\newcommand\GenMC{\toolnamefont{GenMC}}
\newcommand\Saver{\toolnamefont{Saver}}
\newcommand\Nidhugg{\toolnamefont{Nidhugg}}

\definecolor{mygreen}{rgb}{0.05, 0.5, 0.06}

\definecolor{p1bg}{HTML}{ffcc99}
\definecolor{p2bg}{HTML}{cbdcff}
\definecolor{p3bg}{HTML}{ffccff}
\definecolor{p4bg}{HTML}{ccff99}
\def\pa#1{\colorlet{saved}{.}\mbox{\colorbox{p1bg}{\color{saved}#1}}\color{saved}\xspace}
\def\pb#1{\colorlet{saved}{.}\mbox{\colorbox{p2bg}{\color{saved}#1}}\color{saved}\xspace}
\def\pc#1{\colorlet{saved}{.}\mbox{\colorbox{p3bg}{\color{saved}#1}}\color{saved}\xspace}
\def\pd#1{\colorlet{saved}{.}\mbox{\colorbox{p4bg}{\color{saved}#1}}\color{saved}\xspace}
\newcommand\awaiteq[2]{\textbf{await}({#1} $=$ {#2})}
\newcommand\xchgawaiteq[3]{\textbf{xchgawait}({#1} $=$ {#2}, := {#3})}

\newcommand\hbefore[1]{\xrightarrow{\text{hb}}_{#1}}

\newcommand\execution{\textit{E}}

\newcommand\node{\textit{n}}

\newcommand\pre{\textit{pre}}

\newcommand\true{{\it true}}
\newcommand\false{{\it false}}

\algnewcommand\algorithmicswitch{\textbf{switch}}
\algnewcommand\algorithmiccase{\textbf{case}}
\algdef{SE}[SWITCH]{Switch}{EndSwitch}[1]{\algorithmicswitch\ #1\ \algorithmicdo}{\algorithmicend\ \algorithmicswitch}\algdef{SE}[CASE]{Case}{EndCase}[1]{\algorithmiccase\ #1}{\algorithmicend\ \algorithmiccase}\algtext*{EndSwitch}\algtext*{EndCase}

\newcommand{\exploregb}{\mbox{\it Explore}}
\newcommand{\emptytree}{\tuple{\set{\emptyseq},\emptyset}}
\newcommand{\sleepset}{\mbox{\it Sleep}}
\newcommand{\WuT}{\mbox{\it WuT}}
\newcommand{\exseq}{E}
\newcommand{\exseqs}{{\cal E}}
\newcommand{\enables}{\vdash}
  
\newcommand{\enabledafter}[1]{\mbox{\it enabled}(#1)}
\newcommand{\domof}[1]{\mbox{\it dom}(#1)}
\newcommand{\domofafter}[2]{\mbox{\it dom}_{[#1]}(#2)}
\newcommand{\mayreverserace}[1]{\precsim_{#1}}
\newcommand{\notsucc}[2]{\mbox{\notdep}(#1,#2)}
\newcommand{\procof}[1]{\widehat{#1}}
\newcommand{\notdep}{{\it notdep}}
\newcommand{\sleepattr}{\mbox{\it sleep}}
\newcommand{\finalsleep}{\mbox{\it final\_sleep}}

\newcommand{\firsttrans}[2]{\mbox{\it I}_{[#1]}(#2)}
\newcommand{\wfirsttrans}[2]{\mbox{\it WI}_{[#1]}(#2)}
\newcommand{\wut}{\mbox{\it wut}}

\newcommand{\insertseq}[3]{\mbox{\it insert}(#2,#1)}
\newcommand{\infirstseqsop}{\lesssim}
\newcommand{\infirstseqs}[1]{\infirstseqsop_{[#1]}}

\newcommand{\indepafter}[3]{#1 \! \vdash \! #2 \diamondsuit #3}

\newcommand{\subtreeafter}[2]{\mbox{\sl subtree}(#2,#1)}

\newcommand{\nextof}[2]{\mbox{\it next}_{[#1]}(#2)}
\newcommand{\mtequiv}{\simeq}
\newcommand{\mtequivafter}[1]{\simeq_{[#1]}}
\newcommand{\mtclass}[1]{[#1]_{\simeq}}

\newcommand{\happensbefore}[1]{\hbefore{#1}}

\newcommand{\canstop}{\mbox{\it can-stop}}
\newcommand{\didinsert}{\mbox{\it did-insert}}
\renewcommand{\false}{\mbox{\keyword{False}}}
\renewcommand{\true}{\mbox{\keyword{True}}}
\newcommand{\meet}{\vee}
\newcommand{\bigmeet}{\bigvee}

\newcommand{\fpc}{\mbox{FPC}\xspace}
\newcommand{\fpcs}{\mbox{FPCs}\xspace}
\newcommand{\afpc}{a \fpc}
\newcommand{\fpcbefore}[1]{\fpc(\bullet {#1})}
\newcommand{\fpcafter}[1]{\fpc({#1} \bullet)}
\newcommand{\fpcedge}[2]{\fpc({#1}, {#2})}

\newcommand{\pccolor}{purple}
\newcommand{\pclit}[1]{\ensuremath{\mathrm{[}#1\mathrm{]}}\xspace}
\newcommand{\cpclit}[1]{{\color{\pccolor}\pclit{#1}}}
\newcommand{\pcfalse}{\pclit{\false}}

\newcommand{\cpcfalse}{\cpclit{\false}}

\newcommand{\exec}{\execution}
\newcommand{\fpcvar}{\varphi}
\newcommand{\loopvar}{\mathbb{L}}
\newcommand{\progpta}{l}
\newcommand{\progptb}{\progpta'}

\ifextendedver{\newcommand{\orcid}[1]{\orcidlink{#1}}}

\def\aftertablespace{-0.15cm}
\def\afterfigurespace{-0.15cm}
\def\afteralgorithmspace{-0.2cm}

\crefformat{section}{#2\S{}#1#3}
\Crefname{section}{Section}{Sections}
\Crefformat{section}{Section #2#1#3}

\begin{document}

\title{Awaiting for Godot: Stateless Model Checking that
  Avoids Executions where Nothing Happens
  \ifextendedver{\\ {\LARGE (Extended Version with Proofs)}}
}

\author{\IEEEauthorblockN{Bengt Jonsson \orcid{0000-0001-7897-601X}}
\IEEEauthorblockA{Uppsala University, Sweden\\
Email: bengt@it.uu.se}
\and
\IEEEauthorblockN{Magnus Lång \orcid{0000-0003-0984-4229}}
\IEEEauthorblockA{Uppsala University, Sweden\\
Email: magnus.lang@it.uu.se}
\and
\IEEEauthorblockN{Konstantinos Sagonas \orcid{0000-0001-9657-0179}}
\IEEEauthorblockA{Uppsala University, Sweden and NTUA, Greece\\
Email: kostis@it.uu.se}
}

\maketitle

\begin{abstract}
  Stateless Model Checking (SMC) is a verification technique for concurrent
programs that checks for safety violations by exploring all possible thread
schedulings. It is highly effective when coupled with Dynamic Partial Order Reduction (DPOR), which introduces an equivalence on schedulings and need explore only one in each equivalence class.
Even with DPOR, SMC often spends unnecessary effort in exploring loop iterations that are \emph{pure}, i.e., have no effect on the program state.
We present techniques for making SMC with DPOR more effective on programs with
pure loop iterations. The first is a static program analysis to detect
loop purity and an associated program transformation, called
\emph{Partial Loop Purity Elimination}, that inserts \assume statements
to block pure loop iterations.
Subsequently, some of these \assume{s} are turned into \await
statements that completely remove many \assume-blocked executions.
Finally, we present an extension of the standard DPOR equivalence,
obtained by weakening the conflict relation between events.
All these techniques are incorporated into a new DPOR algorithm, \gbalg,
which can handle both \await{s} and the weaker conflict relation, is optimal
in the sense that it explores exactly one execution in each equivalence class,
and can also diagnose livelocks.
Our implementation in \Nidhugg shows that these techniques can significantly
speed up the analysis of concurrent programs that are currently challenging
for SMC tools, both for exploring their complete set of interleavings, but
even for detecting concurrency errors in them.

\end{abstract}

\section{Introduction}

Ensuring correctness of concurrent programs is difficult, since one must
consider all the different ways in which actions of different threads can be
interleaved.
Stateless model checking (SMC)~\cite{Godefroid:popl97} is a fully automatic
technique for finding concurrency bugs (i.e., defects that
arise only under some thread schedulings) and for verifying their absence.
Given a terminating program and fixed input data,
SMC systematically explores the set of
all thread schedulings that are possible during program runs.
A special runtime scheduler drives the SMC exploration by making decisions
on scheduling whenever such choices may affect the interaction between threads.
SMC has been implemented in many tools
(e.g., VeriSoft~\cite{Godefroid:verisoft-journal},
\textsc{Chess}~\cite{MQBBNN:chess}, Concuerror~\cite{Concuerror:ICST13},
\Nidhugg~\cite{tacas15:tso}, rInspect~\cite{DBLP:conf/pldi/ZhangKW15},
\CDSChecker~\cite{NoDe:toplas16}, \RCMC~\cite{KLSV:popl18}, and
\GenMC~\cite{GenMC@CAV-21}), and successfully applied to realistic
programs (e.g.,~\cite{GoHaJa:heartbeat} and~\cite{KoSa:spin17}).

SMC tools typically employ \emph{dynamic partial order reduction}
(DPOR)~\cite{FG:dpor,abdulla2014optimal} to reduce the number of explored
schedulings. DPOR defines an equivalence relation on executions, which
preserves relevant correctness properties, such as reachability of local
states and assertion violations. For correctness, DPOR needs to explore at
least one execution in each equivalence class. We call a DPOR algorithm
\emph{optimal} if it guarantees the exploration of exactly one execution per
equivalence class.

In SMC, loops have to be bounded if they
do not already terminate in a bounded number of iterations.
Loop bounding may in general not preserve assertion failures.
Hence a fairly large loop bound should be used,
but this is often practically infeasible,
and thus loop bounding must strike a balance between these two concerns.
However, for loops whose execution has no global effects, the number of
equivalence classes that need be explored by SMC can be significantly reduced
while still preserving correctness properties, using techniques that we will
present in this paper.

\begin{figure}[t]
  \centering\footnotesize
  \def\linespcadjust{0}\tikzset{style=tight background}
\begin{tikzpicture}[line width=1pt,framed]
      \tikzset{
        full/.style={text=blue},
        scale=0.8,
        every node/.style={scale=0.85},
      }
      \begin{scope}[local bounding box=p box]
        \node(n11)
             {\stmt{\textbf{if}(\xvar{}[0] > \xvar{}[1])}};
        \node(n12) [anchor=north west] at ($(n11.base west)+(5pt,\linespcadjust)$)
             {\stmt{swap(\xvar{}[0], \xvar{}[1]);}};
        \node(n13) [anchor=north west] at ($(n12.base west)+(-5pt,\linespcadjust)$)
             {\stmt{\yvar{} := 1;}};
        \node(n14) [full,anchor=north west] at ($(n13.base west)+(0,\linespcadjust)$)
             {\stmt{\textbf{do} \ \breg{} := \yvar{}}};
        \node(n15) [full,anchor=north west] at ($(n14.base west)+(0,\linespcadjust)$)
             {\stmt{\textbf{while}(\breg{} $\neq$ 2);}};
        \node(n16) [full,anchor=north west] at ($(n15.base west)+(0,\linespcadjust)$)
             {\stmt{\textbf{if}(\xvar{}[0] > \xvar{}[1])}};
        \node(n17) [full,anchor=north west] at ($(n16.base west)+(5pt,\linespcadjust)$)
             {\stmt{swap(\xvar{}[0], \xvar{}[1])}};
      \end{scope}
      \node[anchor=south] at ($(p box.north)+(0,\linespcadjust-5pt)$) {\pa{$p$}};
      \begin{scope}[local bounding box=q box]
        \node(n21) [anchor=north west] at ($(p box.north east)+(12pt,0)$)
             {\stmt{\textbf{do} \areg{} := \yvar{}}};
        \node(n22) [anchor=north west] at ($(n21.base west)+(0,\linespcadjust)$)
             {\stmt{\textbf{while}(\areg{} $\neq$ 1);}};
        \node(n23) [anchor=north west] at ($(n22.base west)+(0,\linespcadjust)$)
             {\stmt{\textbf{if}(\xvar{}[1] > \xvar{}[2])}};
        \node(n24) [anchor=north west] at ($(n23.base west)+(5pt,\linespcadjust)$)
             {\stmt{swap(\xvar{}[1], \xvar{}[2]);}};
        \node(n25) [full,anchor=north west] at ($(n24.base west)+(-5pt,\linespcadjust)$)
             {\stmt{\yvar{} := 2}};
      \end{scope}
      \node[anchor=south] at ($(q box.north)+(0,\linespcadjust-5pt)$) {\pb{$q$}};
      \draw[line width= 0.5pt] ($(p box.north east)+(4.5pt,10pt)$) -- ($(p box.south east)+(4.5pt,2pt)$);
      \draw[line width= 0.5pt] ($(p box.north east)+(7.5pt,10pt)$) -- ($(p box.south east)+(7.5pt,2pt)$);
\end{tikzpicture}
    \label{fig:prodcons:big}
\caption{A concurrent program implementing a sorting network. $p$
    sorts \stmt{\xvar[0]} and \stmt{\xvar[1]}, and then uses \yvar{} to signal
    that \stmt{\xvar[1]} is ready. $q$ waits for \yvar{} to be~1 and then sorts
    \stmt{\xvar[1]} and \stmt{\xvar[2]}, completing one round of bubble sort.
    In the second round, shown in \textcolor{blue}{blue}, $q$ signals that the next ``generation'' of \stmt{\xvar[1]} is
    ready by setting \yvar{} to 2, upon which $p$ finishes the sort by sorting
    \stmt{\xvar[0]} and \stmt{\xvar[1]} again.
    Initially $\yvar = 0$.}
  \label{fig:prodcons}
  \vspace{\afterfigurespace}
\end{figure}

Consider the first round of the program snippet in~\cref{fig:prodcons} (shown in black), where thread $q$ executes a loop that waits for thread $p$
to set the shared variable $\yvar$ to $1$.
A na{\"\i}ve application of SMC with DPOR will explore an unbounded number of
executions, since (in the absence of loop bounding) there is an infinite
number of equivalence classes, one for each number of performed loop
iterations. All iterations of this loop, however, are \emph{pure}, i.e., they
have no effect on the program state. For such loops, a bound of one will
preserve correctness properties. In our example, the \stmt{\textbf{do-while}}
loop of thread $q$ can be rewritten into the sequence of statements
\stmt{\avar{} := \yvar; \textbf{assume}}($\avar = 1$), which will cause the
SMC exploration to permanently block thread $q$ whenever the condition of the
\assume is violated.

Using \assume statements to bound loops causes executions where the condition
of the \assume is violated and its corresponding thread is blocked to be
explored. This happens even if the condition will eventually be satisfied, and
the original
loop will exit, under any fair thread scheduling.
\emph{Assume-blocking} of a thread can occur in many contexts,
each generating an execution that need not be explored.
(We will shortly see this for the example in \cref{fig:prodcons}.)
Furthermore, and perhaps more seriously, this use of \assume{}s prevents SMC
from diagnosing livelocks in which the loop never exits even under fair thread
scheduling.
This is because a blocked execution corresponding to a livelock can also
result from a spurious execution in which the \assume reads a shared variable
before it has been written to by another thread.

Here is where \await statements can lead to further reductions.
An \await loads from a shared variable, but only if the loaded value satisfies
some condition, otherwise it blocks.
In contrast to assume-blocking, \emph{await-blocking} is not permanent
but can be repealed if the condition is later satisfied.
Thereby, executions where blocking occurs by reading ``too early'' are
avoided. Moreover, such executions can be distinguished from livelocks, in
which the condition is not satisfied after some bounded time.
For our example, the rewrite of the \stmt{\textbf{do-while}} loop into an
$\await(\yvar = 1)$ statement results in a program for which SMC would explore
only a single execution in which the \await reads the value written by thread $p$.

Consider now the full program in~\cref{fig:prodcons}, which performs a
concurrent sort of a three-element array using a sorting network. This program
can be scaled to larger arrays for increased available parallelism. Since any
network sorting an array of size $n$ will have at least $\Omega(n \log n)$
occurrences of a code snippet which exchanges two values after exiting a
spinloop, exploring such a program with SMC will explore $\Omega(2^{n \log
  n})$ executions, even after rewriting the spinloops using \assume
statements.
On the other hand, when using \await statements, all executions fall into the
same equivalence class. Thus, an optimal SMC algorithm that can properly handle
\await{}s will explore only one execution, thereby achieving exponential reduction.

In this paper, we present techniques to
\begin{inparaenum}[(i)]
\item automatically transform a program to an intermediate representation that
  uses \await as a primitive, and
\item explore its executions using a provably optimal DPOR algorithm
that is \await aware and also uses a conflict relation between statements
  which is weaker than the standard one.
\end{inparaenum}
We first present a static program analysis technique to detect pure
loop executions and an associated program transformation, called \emph{Partial
Loop Purity (PLP) Elimination}, that inserts \assume statements which are then
turned into \await{}s if preceded by the appropriate load. We prove that PLP
is sound in the sense that it preserves relevant correctness properties,
including local state reachability and assertion failures. We also present and
prove conditions under which PLP is guaranteed to remove all pure executions
of a loop.
Finally, we prove that our new DPOR algorithm \gbalg, which is an extension
of the Optimal-DPOR algorithm of Abdulla et al.~\cite{abdulla2014optimal,optimal-dpor-jacm},
is correct and optimal, also with respect to our weaker conflict relation.

All these techniques are available in \href{https://github.com/nidhugg/nidhugg}{\Nidhugg}, a state-of-the-art SMC tool,
and in the paper's replication package~\cite{godot-artifact@FMCAD-22}.
Our evaluation, using multi-threaded programs which are currently challenging
for most tools, shows that our techniques can achieve significant (and
sometimes exponential) reduction in the total number of executions that need
to be explored. Moreover, they enable detection of concurrency bugs which were
previously out-of-reach for most concurrency testing tools.

\section{Illustration Through Examples}

In this section, we illustrate our contributions through examples.
First, in \cref{sec:spinloops} we show how \assume and \await statements
are inserted. In \cref{sec:alg-by-example} we illustrate how our optimal
DPOR algorithm handles \await statements, and in \cref{sec:inc-by-example}
how it handles the weaker conflict relation in which atomic fetch-and-adds
on the same variable are not conflicting.

We consider programs consisting of a finite set
of threads that share
a finite set of shared variables (\xvar, \yvar, \zvar). A thread has a finite set of local registers (\areg, \breg, \creg), and runs
a deterministic code, built from expressions, atomic statements, and
synchronisation operations, using standard control flow constructs.
Atomic statements read or write to shared variables and local registers,
including atomic read-modify-write operations, such as compare-and-swap and fetch-and-add.
Synchronisation operations include locking a mutex
and joining another thread.
Executions of a program are defined by an interleaving of statements.
We use sequential consistency in this paper, but we note that some weak
memory models (e.g., TSO and PSO) can be modelled by an interleaving-based
semantics, so our work can be extended to DPOR algorithms~\cite{tacas15:tso}
that handle such memory models.
Our loop transformations introduce
\await statements, that take
a conditional expression over a global variable as a parameter and
come in several forms:
simple awaits (\stmt{\awaiteq{\xvar}{0}}),
load-await (\stmt{\areg{} := \awaiteq{\xvar}{0}}),
and exchange-await (\stmt{\areg{} :=  \xchgawaiteq{\xvar}{0}{1}}).
These operations block until their condition is satisfied.

\subsection{Introducing Await Statements}
\label{sec:spinloops}

\begin{figure}[t]
  \centering
  \tikzset{style=tight background}
  \def\linespcadjust{0}\hspace{-3mm}
  \scalebox{0.8}{
\subfloat[][A program with a spinloop.]{
    \begin{tikzpicture}[line width=1pt,framed]
      \begin{scope}[local bounding box=p box]
        \node(n11) {\stmt{\textbf{do} \ \areg{} := x}};
        \node(n12) [anchor=north west] at ($(n11.base west)+(0,\linespcadjust)$)
             {\stmt{\textbf{while}(\areg{} $\neq$ 1);}};
        \node(n13) [anchor=north west] at ($(n12.base west)+(0,\linespcadjust)$)
             {\stmt{\breg{} := \yvar{}}};
      \end{scope}
      \node[anchor=south] at ($(p box.north)+(0,\linespcadjust-5pt)$) {\pa{$p$}};
      \begin{scope}[local bounding box=q box]
        \node(n21) [anchor=north west] at ($(p box.north east)+(6pt,0)$)
             {\stmt{\yvar{} := 42;}};
        \node(n22) [anchor=north west] at ($(n21.base west)+(0,\linespcadjust)$)
             {\stmt{\xvar{} := 1}};
      \end{scope}
      \node[anchor=south] at ($(q box.north)+(0,\linespcadjust-5pt)$) {\pb{$q$}};
      \draw[line width= 0.5pt] ($(p box.north east)+(1.5pt,5pt-\linespcadjust)$) -- ($(p box.south east)+(1.5pt,3pt)$);
      \draw[line width= 0.5pt] ($(p box.north east)+(4.5pt,5pt-\linespcadjust)$) -- ($(p box.south east)+(4.5pt,3pt)$);
    \end{tikzpicture}
    \label{fig:spinwait:original}
  }\subfloat[][$p$ rewritten with \assume.]{
    \begin{tikzpicture}[line width=1pt,framed]
      \begin{scope}[local bounding box=p box]
        \node(n11) {\stmt{\areg{} := \xvar{};}};
        \node(n12) [anchor=north west] at ($(n11.base west)+(0,\linespcadjust)$)
             {\stmt{\textbf{assume}(\areg{} $=$ 1);}};
        \node(n13) [anchor=north west] at ($(n12.base west)+(0,\linespcadjust)$)
             {\stmt{\breg{} := \yvar{}}};
      \end{scope}
      \node[anchor=south] at ($(p box.north)+(0,\linespcadjust-5pt)$) {\pa{$p$}};
\end{tikzpicture}
    \label{fig:spinwait:assume}
  }\subfloat[][$p$ rewritten with \await.]{
    \centering
    \begin{tikzpicture}[line width=1pt,framed]
      \begin{scope}[local bounding box=p box]
        \node(n11) {\stmt{\awaiteq{\xvar}{1};\space}};
        \node(n12) [anchor=north west] at ($(n11.base west)+(0,\linespcadjust)$)
             {\stmt{\breg{} := \yvar{}}};
      \end{scope}
      \node[anchor=south] at ($(p box.north)+(0,\linespcadjust-5pt)$) {\pa{$p$}};
\end{tikzpicture}
    \label{fig:spinwait:await}
  }
  }
  \caption{Multi-threaded program illustrating the rewrites;
    initially, $\xvar = \yvar = 0$. For (b) and (c), $q$ is the same as in (a).}
  \label{fig:spinwait}
  \vspace{\afterfigurespace}
\end{figure}
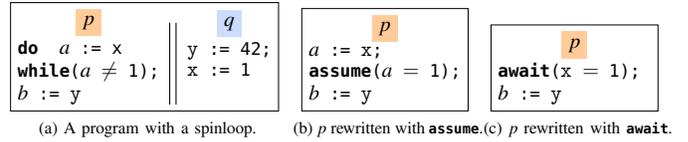

Let us show an example of how loops are transformed by introducing
\assume and \await statements.
Consider the loop in \cref{fig:spinwait:original}.
There, thread $p$ executes a spinloop, waiting for thread $q$ to set the shared
variable~$\xvar$.
Each iteration of this loop, in which the value loaded into $\avar$ is different from $1$,
is pure, i.e., it does not modify shared variables, nor any local register that may
be used after the end of the loop. Therefore an \assume statement is introduced at
the point where the thread can distinguish pure executions from impure ones, i.e., after
$\areg$ has been loaded.
The result of such a rewrite is shown in \cref{fig:spinwait:assume}.
This program has two traces, one in which the \assume succeeds, representing the executions
in which the original loop terminates, and one where thread $p$ gets \emph{assume-blocked}.
The latter trace will exist even in the case where the original loop is guaranteed to
terminate under a fair scheduler.
This problem is remedied by replacing the load into $\areg$ and the following
\assume statement by an \await with a test on the shared variable from which
$\areg$ reads.
Such a rewrite results in the program in \cref{fig:spinwait:await}.
In this case, the \await statement may permanently block only if the original
loop can livelock under fair scheduling.
In our simple example, the rewritten program has only a single trace, since
the original loop is guaranteed to terminate and can be replaced by the \await.
Programs with more complex loops (e.g., loops that are pure only along a
subset of their paths) are also handled by our program transformation
(\cref{sec:plp}), but the loop is not eliminated when \assume{}s or \await{}s
are introduced.

\subsection{\gbalg by Example}
\label{sec:alg-by-example}

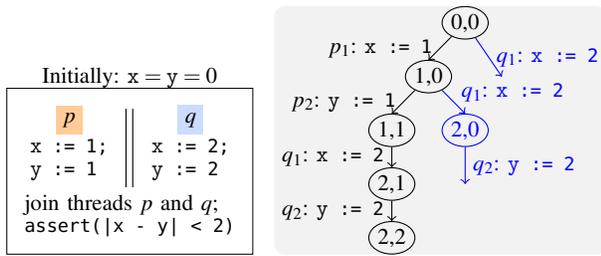
\begin{figure}[t]
  \centering
  \scalebox{0.85}{
  \begin{minipage}[b]{0.45\linewidth}
    \centering
    \def\linespcadjust{0}\scalebox{1}{Initially: \mbox{$\xvar = \yvar = 0$}}
    \begin{tikzpicture}[line width=1pt,framed,
        scale=1,
        every node/.style={scale=1},
      ]
      \begin{scope}[local bounding box=p box]
        \node(n11) {\stmt{\xvar{} := 1;}};
        \node(n12) [anchor=north west] at ($(n11.base west)+(0,\linespcadjust)$)
             {\stmt{\yvar{} := 1}};
      \end{scope}
      \node[anchor=south] at ($(p box.north)+(0,\linespcadjust-5pt)$) {\pa{$p$}};
      \begin{scope}[local bounding box=q box]
        \node(n21) [anchor=north west] at ($(p box.north east)+(12pt,0.5pt)$)
             {\stmt{\xvar{} := 2;}};
        \node(n22) [anchor=north west] at ($(n21.base west)+(0,\linespcadjust)$)
             {\stmt{\yvar{} := 2}};
      \end{scope}
      \node[anchor=south] at ($(q box.north)+(0,\linespcadjust-5pt)$) {\pb{$q$}};
      \draw[line width= 0.5pt] ($(p box.north east)+(4.5pt,10pt)$) -- ($(p box.south east)+(4.5pt,0)$);
      \draw[line width= 0.5pt] ($(p box.north east)+(7.5pt,10pt)$) -- ($(p box.south east)+(7.5pt,0)$);

      \node(nj1) [anchor=north west] at ($(n12.base west)+(-2.5pt,-5pt)$) {join threads $p$ and $q$;};
      \node(nj2) [anchor=north west] at ($(nj1.base west)+(0,\linespcadjust)$) {{\tt assert(|x - y| < 2)}};
    \end{tikzpicture}
  \end{minipage}
  }
  \begin{minipage}[b]{0.24\textwidth}
    \centering
    \begin{tikzpicture}[
        level distance=1.1cm,
        st/.style={ellipse,inner sep=1pt,draw},
        child anchor=north,
        scale=0.65,
        every node/.style={scale=0.85},
      ]
      \node[st] {0,0}
      child { node[st] {1,0}
        child { node[st] {1,1}
          child { node[st] {2,1}
            child { node[st] {2,2}
              edge from parent [->] node [left] {$q_2$: \stmt{\yvar{} := 2}}
            }
            edge from parent [->] node [left] {$q_1$: \stmt{\xvar{} := 2}}
          }
          edge from parent [->] node [left] {$p_2$: \stmt{\yvar{} := 1}}
        }
        child { node[st,color=blue] {2,0}
          child {
            edge from parent [->,color=blue] node [right] {$q_2$: \stmt{\yvar{} := 2}}
          }
          edge from parent [->,color=blue] node [above right=-3pt and 0pt] {$q_1$: \stmt{\xvar{} := 2}}
        }
        edge from parent [->] node [left] {$p_1$: \stmt{\xvar{} := 1}}
      }
      child {
        edge from parent [->,color=blue] node [right] {$q_1$: \stmt{\xvar{} := 2}}
      };
      \tikzwrapfigbg
    \end{tikzpicture}
\end{minipage}
\caption{Program with a correctness assertion, and
    execution trees with the first scheduling of the program; nodes show the
    values of variables $\xvar$ and~$\yvar$.}
  \label{fig:example1}
  \vspace{\afterfigurespace}
\end{figure}

DPOR algorithms are based on regarding executions as equivalent if they induce the same ordering between
executions of conflicting statements. The standard conflict relation regards two accesses to the same variable as
conflicting if at least one is a write.
We begin by illustrating the Optimal-DPOR algorithm~\cite{optimal-dpor-jacm}
on the simple program in \cref{fig:example1}. There two threads, $p$ and $q$,
write to two shared variables $\xvar$ and $\yvar$ in sequence.
Optimal-DPOR starts by exploring an arbitrary interleaved execution of
the program. Assume it is $p_1.p_2.q_1.q_2$ as shown in \cref{fig:example1}
(we will denote executions by sequences of thread identifiers, possibly subscripted
by sequence numbers).
Each explored execution is then analysed
to find \emph{races}, i.e., pairs of conflicting events that are adjacent in the
happens-before order induced by the conflict relation.
(An \emph{event} is a particular execution step of a thread in an execution.)
Our first execution contains two races, $(p_1,q_1)$ and $(p_2,q_2)$.
For each race, Optimal-DPOR creates a so-called \emph{wakeup sequence},
i.e., a sequence which continues the analysed execution up to the first event in a way which reaches
the second event instead of the first event.
For the first race, the wakeup sequence is $q_1$, and for the second race,
it is 
$p_1.q_1.q_2$.
The wakeup sequences are inserted as new branches just before the first event of the
corresponding race, thereby gradually building a tree consisting of the explored executions and
added wakeup sequences.
The execution tree after the first execution is shown in
\cref{fig:example1}.

After processing the first execution, Optimal-DPOR then picks the leftmost unexplored leaf
in the tree, and extends it arbitrarily to a full execution, in which races are analysed, etc.
As the algorithm backtracks, it deletes the nodes it backtracks from in the execution
tree.
The second execution has two races, $(p_1,q_1)$ as well as $(p_2,q_2)$. However, the corresponding wakeup sequences will result in executions that are redundant, i.e.,
equivalent to already inserted ones, so no further insertion takes place.
The algorithm proceeds in this way until there are no more unexplored leafs corresponding to wakeup
sequences. In total, there are four executions explored by Optimal-DPOR,
corresponding to the four possible final valuations of \xvar{} and \yvar{}.

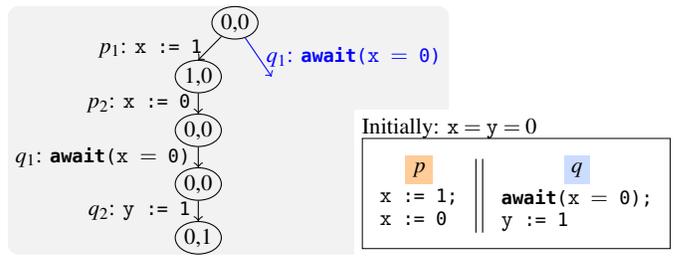
\begin{figure}[t]
\begin{minipage}[b]{\columnwidth}
\begin{tikzpicture}[
        scale=0.65,
        level distance=1.1cm,
        st/.style={ellipse,inner sep=1pt,draw},
        child anchor=north,
        every node/.style={scale=0.85},
      ]
      \node[st] {0,0}
      child { node[st] {1,0}
        child { node[st] {0,0}
          child { node[st] {0,0}
            child { node[st] {0,1}
              edge from parent [->] node [left] {$q_2$: \stmt{\yvar{} := 1}}
            }
            edge from parent [->] node [left] {$q_1$: \stmt{\awaiteq{\xvar}{0}}}
          }
          edge from parent [->] node [left] {$p_2$: \stmt{\xvar{} := 0}}
        }
        edge from parent [->] node [left] {$p_1$: \stmt{\xvar{} := 1}}
      }
      child {
        edge from parent [->,color=blue] node [right] {$q_1$: \stmt{\awaiteq{\xvar}{0}}}
      };
      \node[below right=3.3em and 4.5em,fill=white,overlay]  {
        \begin{minipage}[b]{0.28\textwidth}
          \centering
          \def\linespcadjust{0}\scalebox{1}{Initially: \mbox{$\xvar = \yvar = 0$}}
          \begin{tikzpicture}[line width=1pt,framed,
              scale=1,
              every node/.style={scale=1},
            ]
            \begin{scope}[local bounding box=p box]
              \node(n11) {\stmt{\xvar{} := 1;}};
              \node(n12) [anchor=north west] at ($(n11.base west)+(0,\linespcadjust)$)
                   {\stmt{\xvar{} := 0}};
            \end{scope}
            \node[anchor=south] at ($(p box.north)+(0,\linespcadjust-5pt)$) {\pa{$p$}};
            \begin{scope}[local bounding box=q box]
              \node(n21) [anchor=north west] at ($(p box.north east)+(12pt,0.5pt)$)
                   {\stmt{\awaiteq{\xvar}{0};}};
                   \node(n22) [anchor=north west] at ($(n21.base west)+(0,\linespcadjust)$)
                        {\stmt{\yvar{} := 1}};
            \end{scope}
            \node[anchor=south] at ($(q box.north)+(0,\linespcadjust-5pt)$) {\pb{$q$}};
            \draw[line width= 0.5pt] ($(p box.north east)+(4.5pt,10pt)$) -- ($(p box.south east)+(4.5pt,0)$);
            \draw[line width= 0.5pt] ($(p box.north east)+(7.5pt,10pt)$) -- ($(p box.south east)+(7.5pt,0)$);
          \end{tikzpicture}
        \end{minipage}
      };
      \tikzwrapfigbg
    \end{tikzpicture}
  \end{minipage}
  \caption{Exploration of a program with an \await with two satisfying writes.}
  \label{fig:awaitex}
  \vspace{\afterfigurespace}
\end{figure}

Let us now look at how \gbalg extends Optimal-DPOR to work for programs with \await{}s.
Consider the program in \cref{fig:awaitex}. There, $p$ writes to the global
variable \xvar, first updating it to $1$, and then back to $0$. Assume that the
first execution is $p_1.p_2.q_1.q_2$.
The analysis of races performed by Optimal-DPOR must now be extended to consider that
\await statements are sometimes blocked. First, the conflict between
$p_2$ with $q_1$ will not be handled like a race, since $q_1$ is blocked just before $p_2$.
Therefore, we find the closest preceding point in the execution at which $q_1$ is
not blocked, which in this case is at the beginning.
We then construct the wakeup sequence $q_1$ and insert it at the beginning;
cf. \cref{fig:awaitex}. Since this program only has two traces, \gbalg
will terminate after exploring the second execution.

\subsection{Handling Atomic Fetch-and-Add Instructions in DPOR}
\label{sec:inc-by-example}

To reduce the number of equivalence classes that need be explored by a DPOR algorithm,
one can weaken the standard conflict relation between statements by considering
two atomic fetch-and-add (FAA) statements on the same variable as non-conflicting
if the loaded values are afterwards unused.
In the absence of \await statements, many existing DPOR algorithms like Optimal-DPOR handle this definition without modification.
However, this weakening has a subtle interaction with \await statements
that must be handled by \gbalg.

Consider the program in \cref{fig:optimalgb-example}. In this
program, three threads, $p$, $q$, and $r$, add atomically to the shared
variable $\xvar$, and a thread $s$ awaits $\xvar$ having the value $3$. We
assume that DPOR considers the FAA statements $p_1$, $q_1$, and
$r_1$ to be non-conflicting, but conflicting with the statement $s_1$, should it
execute.

\begin{figure}[t]

\begin{tikzpicture}[
        scale=0.65,
        level distance=1.1cm,
        st/.style={ellipse,inner sep=1pt,draw},
        child anchor=north,
        every node/.style={scale=0.85},
      ]
      \node[st] {0,0}
      child { node[st] {1,0}
        child { node[st] {2,0}
          child { node[st] {5,0}
            edge from parent [->] node [left] {$r_1$: \stmt{x+:=3}}
          }
          edge from parent [->] node [left] {$q_1$: \stmt{x+:=1}}
        }
        edge from parent [->] node [left] {$p_1$: \stmt{x+:=1}}
      }
      child { node[st,color=blue] {3,0}
        child {
          edge from parent [->,color=blue] node [right] {$s_1$: \stmt{\awaiteq{\xvar}{3}}}
        }
        edge from parent [->,color=blue] node [right] {$r_1$: \stmt{x+:=3}}
      };
      \node[below right=3.9em and 1.7em,fill=white,overlay]  {
        \begin{minipage}[b]{0.35\textwidth}
          \centering
          \small
          \def\linespcadjust{0}\begin{tikzpicture}[line width=1pt,framed,
              scale=1,
              every node/.style={scale=1},
              style=tight background,
            ]
            \begin{scope}[local bounding box=p box]
              \node(n11) {\stmt{\xvar{} +:= 1}};
            \end{scope}
            \node[anchor=south] at ($(p box.north)+(0,\linespcadjust-5pt)$) {\pa{$p$}};
            \begin{scope}[local bounding box=q box]
              \node(n21) [anchor=north west] at ($(p box.north east)+(10pt,0.5pt)$)
                   {\stmt{\xvar{} +:= 1}};
            \end{scope}
            \node[anchor=south] at ($(q box.north)+(0,\linespcadjust-5pt)$) {\pb{$q$}};
            \draw[line width= 0.5pt] ($(p box.north east)+(3.5pt,10pt)$) -- ($(p box.south east)+(3.5pt,0)$);
            \draw[line width= 0.5pt] ($(p box.north east)+(6.5pt,10pt)$) -- ($(p box.south east)+(6.5pt,0)$);
            \begin{scope}[local bounding box=r box]
              \node(n31) [anchor=north west] at ($(q box.north east)+(10pt,0.5pt)$)
                   {\stmt{\xvar{} +:= 3}};
            \end{scope}
            \node[anchor=south] at ($(r box.north)+(0,\linespcadjust-5pt)$) {\pc{$r$}};
            \draw[line width= 0.5pt] ($(q box.north east)+(3.5pt,10pt)$) -- ($(q box.south east)+(3.5pt,0)$);
            \draw[line width= 0.5pt] ($(q box.north east)+(6.5pt,10pt)$) -- ($(q box.south east)+(6.5pt,0)$);
            \begin{scope}[local bounding box=s box]
              \node(n41) [anchor=north west] at ($(r box.north east)+(10pt,1pt)$)
                   {\stmt{\awaiteq{\xvar}{3};}};
                   \node(n42) [anchor=north west] at ($(n41.base west)+(0,\linespcadjust)$)
                        {\stmt{\yvar{} := 1}};
            \end{scope}
            \node[anchor=south] at ($(s box.north)+(0,\linespcadjust-5pt)$) {\pd{$s$}};
            \draw[line width= 0.5pt] ($(r box.north east)+(3.5pt,10pt)$) -- ($(r box.south east)+(3.5pt,0)$);
            \draw[line width= 0.5pt] ($(r box.north east)+(6.5pt,10pt)$) -- ($(r box.south east)+(6.5pt,0)$);
          \end{tikzpicture}
        \end{minipage}
      };
      \tikzwrapfigbg
    \end{tikzpicture}
    \vspace{0.6em} \caption{Exploration of a program with fetch-and-adds. Initially, \mbox{$\xvar = \yvar = 0$}.}
  \label{fig:optimalgb-example}
  \vspace{\afterfigurespace}
\end{figure}
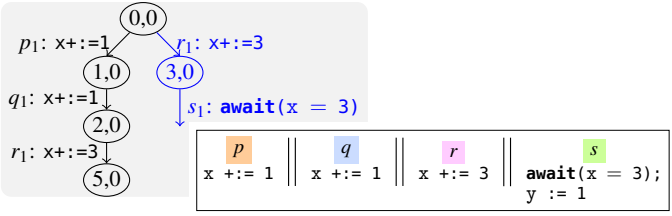

Assume that the first explored execution is $p_1.q_1.r_1$. From this point,
we cannot substitute $s_1$ for either of $p_1$, $q_1$, or $r_1$, as $s_1$ is not
enabled after any of $q_1.r_1$, $p_1.r_1$ or $p_1.q_1$, respectively. Yet, there
is another execution in which $s_1$ is enabled. In order to construct
this execution, we must not only schedule $s_1$ before one of the other events,
but before two, both of $p_1$ and $q_1$, so that only $r_1$ remains. Then, we
could construct the wakeup sequence $r_1.s_1$.
In general, \gbalg may need to reorder the sequence of independent FAAs
that precede  an \await statement and select a subsequence of them, in order
to unblock the \await statement. This can be done in several ways,
and \gbalg is optimised to avoid enumerating all of them.
In \cref{sec:algorithm-code}, we will see how.

\section{Partial Loop Purity Elimination}
\label{sec:plp}
In this section, we describe \emph{Partial Loop Purity Elimination},
a technique that prevents SMC from exploring executions
with pure loop iterations. It consists of
\begin{inparaenum}[(1)]
\item a static analysis technique which annotates programs with conditions
  under which a loop will execute a pure iteration, and
\item a program transformation which inserts \assume statements based on
  the analysis.
\end{inparaenum}

We consider loops consisting of a set of basic blocks, with a single
header block. Each basic block contains a sequence of program statements.
Blocks are connected via edges, labelled by conditions.
We also consider program representations on Static Single Assignment (SSA) form,
which means that each register is assigned by exactly one statement.
Thus, a register uniquely identifies the statement that assigns to it.
When the value of a register in one block depends on which predecessor block was
executed, this is expressed using a \emph{phi node}. For example, in a block $C$
with predecessors $A$ and $B$ containing registers $\areg$ and $\breg$,
respectively, the statement $\creg := \phi(A:\areg, B:\breg)$ defines the
register $\creg$ to get the value of $\areg$ when the previous basic block was
$A$ and of $\breg$ when the previous block was $B$.

An execution of a loop iteration is \emph{pure} if the execution starts and
ends at the header of the loop, and during the iteration
\begin{inparaenum}[(i)]
\item no modification of a global variable is performed,
\item nor of any local variable that may be used after the end of the iteration, and
\item no \emph{internal} (not to the header) backedge is taken.
\end{inparaenum}
In SSA form,
modification of local variables can be inferred from
the phi nodes in
the header. If such a phi node uses a different value
on the backedge
to the header
than when first entering, then the loop iteration modified a local variable that is used
on some path after the iteration\newcommand{\impureheader}{impure\xspace}, and we call the header \emph{\impureheader} along the backedge.
Our definition considers executions that complete inner loop iterations to be non-pure.
However, our PLP transformation will block inner loops from completing pure iterations.

A register \areg{} \emph{reaches} a program point $\progpta$ if
all paths to $\progpta$ pass \areg's definition.
\newcommand{\deftrue}{defined-true\xspace}During a loop execution,
  we say that an expression over registers
  is \emph{\deftrue} at some program point~$\progpta$ in the loop,
  if the expression evaluates to true under
  \begin{inparaenum}[(i)]
  \item the current valuation of registers that were assigned either outside
    the loop or during the current loop iteration, and
  \item any valuation of all other registers.
  \end{inparaenum}
We now define a central concept; that of the Forward Purity Condition.

\begin{definition}[Forward Purity Condition]
  Let $\progpta$ be a program point in a loop.
Then, a \emph{Forward Purity Condition} (\fpc) at
  $\progpta$ is an expression in Disjunctive Normal Form
  over the registers such that
  if an execution, without leaving the loop or taking an internal backedge,
  proceeds to a program point $\progptb$,
  at which the expression is \deftrue, then
\begin{enumerate}[(i)]
  \item the execution from $\progptb$ will reach the loop header
    without taking an internal backedge, and
  \item the execution from $\progpta$ to the loop header will not modify any global
    variables nor any local variable that may be used after execution has reached
    the loop header. \qed
\end{enumerate}
\end{definition}
We will denote \afpc with brackets, for example \pclit{\creg > 42} or \pcfalse.
A \emph{purity condition} (PC) of a loop is \afpc of the loop at the beginning
of its header. Thus, whenever a loop iteration passes a program point where the PC is
\deftrue, and has not taken an internal backedge, then that iteration is pure.

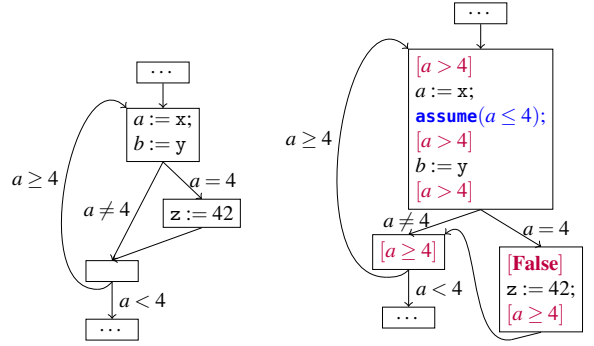
\begin{figure}[t]
  \centering
  \subfloat[][A loop with non-purity\\ and conditional branches.\label{fig:plp-branched-before}]{
    \scalebox{0.93}{
    \begin{tikzpicture}
      \tikzset{
        minimum height=1em,
        minimum width=2.5em,
        node distance=1.5em,
        bb/.style={draw, rectangle, align=left},
        every node/.style={scale=0.85},
      }
      \node[bb] (before) {\dots};
      \node[bb,below=10pt of before] (header) {$\areg := \xvar$;\\ $\breg := \yvar$};
      \draw[->] (before) -> (header);
      \node[bb,below right=1.5em and -1.5em of header] (impure) {$\zvar := 42$};
      \path (header.south) edge[->] node[midway,right] {$\areg = 4$} (impure.north);
      \node[bb,below left=4em and -0.5em of header] (cond) {};
      \path (impure.south) edge[->] (cond.north);
      \path (header.south) edge[->] node[left] {$\areg \neq 4$} (cond.north);
      \draw (cond.south) to[out=225,in=0] ++(-6pt,-3pt)
      edge[->,out=180,in=135] node[midway,left] {$\areg \geq 4$}
      (header.north west);
      \node[bb,below=of cond] (after) {\dots};
      \path (cond) edge[->] node[right] {$\areg < 4$} (after);
\end{tikzpicture}
    }
  }
  \subfloat[][The loop annotated with {\color{\pccolor}\fpcs} and with the {\color{blue}\assume} that is inserted.\label{fig:plp-branched-after}]{
    \scalebox{0.93}{
    \begin{tikzpicture}
      \tikzset{
        minimum height=1em,
        minimum width=2.5em,
        node distance=1.5em,
        bb/.style={draw, rectangle, align=left},
        every node/.style={scale=0.85},
        bezier bounding box=true,
      }
      \node[bb] (before) {\dots};
      \node[bb,below=10pt of before] (header) {\cpclit{\areg > 4}\\
        $\areg := \xvar$;\\
        \textcolor{blue}{$\assume(\areg \leq 4)$;}\\
        \cpclit{\areg > 4}\\
        $\breg := \yvar$\\
        \cpclit{\areg > 4}};
      \draw[->] (before) -> (header);
      \node[bb,below right=1.5em and -2.2em of header] (impure) {\cpcfalse\\
        $\zvar := 42$;\\
        \cpclit{\areg \geq 4}};
      \path (header.south) edge[->] node[midway,right=2pt]
            {$\areg = 4$} (impure.north);
      \node[bb,below left=1em and -1.5em of header] (cond)
            {$\cpclit{\areg \geq 4}$};
      \draw (impure.south) to ++(-12pt,-3pt) edge[->,out=190,in=40] (cond.north
      east);
      \path (header.south) edge[->] node[left=3pt] {$\areg \neq 4$} (cond.north);
\draw (cond.south) to[out=225,in=0] ++(-9pt,-3pt)
      edge[->,out=180,in=135] node[midway,left] {$\areg \geq 4$}
      (header.north west);
\node[bb,below=of cond] (after) {\dots};
      \path (cond) edge[->] node[right] {$\areg < 4$} (after);
\end{tikzpicture}
    }
  }
  \caption{Program snippet illustrating the concepts of the PLP transformation.}
  \label{fig:plp-branched}
  \vspace{\afterfigurespace}
\end{figure}

We illustrate these concepts for the program snippet in
\cref{fig:plp-branched-before}. In it, the loop loads $\xvar$ and $\yvar$ into
registers $\avar$ and~$\bvar$, then branches on the value of $\areg$, and along
the path where $\areg
= 4$, there is a write to $\zvar$. Since a write to a global variable is
non-pure, the loop is not pure whenever $\areg = 4$.
The two paths converge in a common block where a loop
condition ($\areg \geq 4$) is checked. This loop is pure if
\begin{inparaenum}[(i)]
\item it takes the backedge, i.e., $\areg \geq 4$ holds, and
\item the write to $\zvar$ is not performed, i.e., $\areg \neq 4$ also
  holds.
\end{inparaenum}
The conjunction of these conditions, $\areg > 4$, becomes a purity condition
for the entire loop.
We thereafter insert an \assume with the negation of a disjunct of the PC
at the earliest point that it is \deftrue, i.e.,
after the load of $\xvar$, shown in blue in \cref{fig:plp-branched-after}.

Let us now describe the analysis stage for computing purity conditions.
Its first step is to compute \fpcs at all points in the loop.
Intuitively, the \fpc at a point $\progpta$
is a disjunction $c_1 \meet \cdots \meet c_n$, where each $c_i$ is a (forward)
path condition for reaching the header via a pure execution from $\progpta$.
We compute \fpcs by backwards propagation through statements and basic blocks.
Let $\fpcafter{s}$ be the \fpc immediately after statement $s$,
let $\fpcbefore{s}$ be the \fpc immediately before statement $s$,
let $\fpcbefore{B}$ be the \fpc at the beginning of block $B$, and
let $\fpcafter{B}$ be the \fpc at the end of block $B$.

For each statement $s$, we compute $\fpcbefore{s}$ as $\fpcafter{s} \land g$,
where $g$ is the condition under which $s$ does not update a global variable.
For instance, $g$ is $\false$ for stores, $\true$ for loads,
$\areg{} = 0$ for an atomic add of form \stmt{\xvar{} +:= \areg{}}, 
$\mbox{\areg{}} = \mbox{\breg{}}$ for
an atomic exchange of form \stmt{\breg{} := xchg(\xvar{},\areg{})},
  and
$\mbox{\creg{}} = 1$ for an atomic compare-exchange of form \stmt{\creg{} := cmpxchg(\xvar{},\areg{},\breg{})}.

\fpcs for basic blocks are computed as follows. First, for an
edge with condition $g$ from a block $A$ in the loop to a block $B$,
let $\fpcedge AB$ be the \fpc along that edge, defined as follows;
\begin{itemize}
  \item if $B$  is outside the loop, then $\fpcedge AB = \pcfalse$,
  \item if $B$ is the header block, then
    if $B$ is \impureheader along $(A,B)$, then $\fpcedge AB = \pcfalse$, otherwise $\fpcedge AB = \pclit{g}$.
\item if $B$ is inside the loop, then $\fpcedge AB = \pcfalse$ if the edge
    from $A$ to $B$ is an internal backedge $(A,B)$, otherwise $\fpcedge AB = \pclit{\fpcbefore{B} \wedge g}$,
\end{itemize}

We propagate \fpcs backwards through basic blocks by the above rules for statements.
We then compute the \fpc at the end of a block $A$ with outgoing arcs to $B_1,
\ldots , B_k$ as $\fpcafter{A} = \bigmeet_{i=1}^k \fpcedge{A}{B_i}$.
We can thereafter calculate \fpcs for basic blocks by starting from the edges
that leave the loop or go back to its header.
Cycles in the control flow graph are no issue, since the \fpc of a backedge
$(A,B)$ does not depend on~$B$.
In \cref{fig:plp-branched-after}, we can see the \fpcs computed by the analysis
on the example.

After the analysis, we insert \assume statements. Given a purity condition of
form $c_1 \meet c_2 \meet \cdots \meet c_n$, for each $c_i$ we insert an
$\assume(\neg c_i)$ at the earliest point
that is textually after the definitions of
all registers in $c_i$.
For registers that do not reach the insertion location, arbitrary values can
be used when execution does not pass their definitions.
Moreover, if any memory access along the path corresponding to $c_i$
cannot be statically determined not to segfault, we must not insert $c_i$
before that memory access. For this purpose, we associate an optional ``earliest
insertion point'' with every $c_i$ in each \fpc computed by the analysis.
Finally, to exclude paths that took some internal backedge, a ``took
internal backedge'' boolean register is introduced, computed by phi-nodes,
and included in the conjunction $c_i$.

\Cref{thm:plp-safety}, whose proof appears in
\iffwithappendix{\cref{app:plp-proof}}{the
  extended version~\cite{godot-extended@arXiv-22} of this paper},
states two essential properties of PLP. These properties intuitively say that
PLP removes pure executions while preserving relevant correctness properties.
If $\thstate$ is a local state occurring in a loop $\loopvar$ of a thread
$p$, we say that $\loopvar$ is \emph{unavoidably pure} from $\thstate$ to
denote that whenever thread $p$ is in local state $\thstate$ during an
execution, then $p$ is in the process of completing a pure iteration of
$\loopvar$.

\begin{theorem}
  \label{thm:plp-safety}
  \label{thm:plp-optimality}
  Let $\prog'$ be the program resulting from applying PLP to $\prog$. Then
  $\prog'$ satisfies the following properties.
  \begin{enumerate}
  \item \emph{Local State Preservation:} \label{thm:plp-safety-prop}
    each local state $\thstate$ of a thread $p$ which is reachable in $\prog$
    is also reachable in $\prog'$, provided no loop of $p$ is unavoidably pure
    from $\thstate$.
  \item \emph{Pure Loop Elimination:}\label{thm:plp-optimality-prop}
    no execution of $\prog'$ exhibits a completed pure loop iteration of some
    thread.
  \end{enumerate}
\end{theorem}
We remark that in the definition of pure loop iterations, we assume possibly
conservative characterisations of
``global variable'' and ``local variable that may be used after the end of the iteration''
that can be determined by a standard syntactical analysis of the program, and hence
used in the PLP analysis.

\section{The \gbalg Algorithm}
\label{sec:algorithm}

In this section, we present \emph{\gbalg}, a DPOR algorithm for programs with \await statements, which is
both correct and optimal. Given a terminating program on given input,
it explores exactly one maximal execution in each equivalence class induced by the
equivalence relation $\mtequiv$.

\subsection{Happens-Before Ordering and Equivalence}
DPOR algorithms are based on a partial order on the events in each execution.
Given an execution $\exseq$ of a program $\prog$,
an \emph{event} of $\exseq$ is a  particular execution step  by a single thread;
the $i$'th event by thread $p$ is identified by the tuple $\tuple{p,i}$, and
$\procof{\event}$ denotes the thread $p$ of an event $\event = \tuple{p,i}$.
Let $\domof{\exseq}$  denote the set of events in $\exseq$.
We define a {\em happens-before relation} on $\domof{\exseq}$, denoted $\happensbefore{\exseq}$,
as the smallest transitive relation such that
$\event \happensbefore{\exseq} \event'$ if
$\event$ occurs before $\event'$ in $\exseq$, and either
\begin{itemize}
\item[(i)]
  $\event$ and $\event'$ are performed by the same thread,
  $\event$ spawns the thread which performs $\event'$, or
  $\event'$ joins the thread which performs $\event$, or
\item[(ii)]
  $\event$ and $\event'$ access a common shared variable \xvar, at least one of them writes to \xvar, and they are not both atomic fetch-and-add operations.
\end{itemize}
Note that the last condition makes atomic fetch-and-add operations on the same shared variable
independent.
It follows that $\happensbefore{\exseq}$ is a partial order on $\domof{\exseq}$.
We define two executions, $\exseq$ and $\exseq'$, as equivalent, denoted $\exseq \mtequiv \exseq'$,
if they induce the same happens-before relation on the same set of events,
(i.e., $\domof{\exseq} = \domof{\exseq'}$ and $\happensbefore{\exseq} = \happensbefore{\exseq'}$).
If $\exseq \mtequiv \exseq'$, then all variables are modified by the
same sequence of statements, implying that each thread runs through the same
sequence of local states in $\exseq$ and $\exseq'$.

\subsection{The Working of the \gbalg Algorithm}
\label{sec:algorithm-code}

\gbalg is shown in \cref{alg:dpor-alg-blocking2}.
It performs a depth-first exploration of executions using the recursive procedure
$\exploregb(\exseq)$, where $\exseq$ is the currently explored execution,
which can also be interpreted as the stack of the depth-first exploration.
In addition, for each prefix $\exseq'$ of $\exseq$, the algorithm maintains
\begin{itemize}
\item a \emph{sleep set} $\sleepattr(\exseq')$, i.e., a set of threads that should not be explored from $\exseq'$,
  for the reason that each extension of form $\exseq'.p$ for $p \in \sleepattr(\exseq')$ is equivalent to a
  previously explored sequence,
\item a \emph{wakeup tree} $\wut(\exseq')$, i.e., an ordered tree $\tuple{B,\prec}$, where
$B$ is a prefix-closed set of sequences, whose leaves are called \emph{wakeup sequences}, and
$\prec$ is the order in which sequences were added to $\wut(\exseq')$.
For each $w \in B$ the sequence $\exseq'.w$
  will be explored during the call $\exploregb(\exseq')$ in the order given by $\prec$.
\end{itemize}
All previously explored sequences together with the current wakeup tree (i.e.,
all sequences of form $\exseq'.w$ for $w \in \wut(\exseq')$ and a prefix $\exseq'$ of $\exseq$)
form the current execution tree, denoted $\exseqs$.
The branches of $\exseqs$ are ordered
by the order in which they were added to the tree.
Note that the recursive call to
$\exploregb(\exseq)$ may insert into $\wut(\exseq')$ for prefixes $\exseq'$ of $\exseq$.

Let $v \setminus p$ denote the sequence $v$ with the first occurrence of an
event by thread $p$ (if any) removed.
Let $\nextof{\exseq}{p}$ denote the next event performed by thread $p$ after $\exseq$.
Two important concepts are \emph{races} and \emph{weak initials}.

\begin{definition}[Non-Blocking Races]
\label{def:races}
Let $\event, \event'$ be two events in different threads in an execution $\exseq$, where $\event$ occurs before $\event'$.
Then  $\event$ and $\event'$ are in
a \emph{non-blocking race}, denoted $\event \mayreverserace{\exseq} \event'$, if
\begin{inparaenum}[(i)]
  \item $\event$ and $\event'$ are adjacent in $\hbefore{\exseq}$
    (i.e., $e \hbefore{\exseq} e'$, and for no other event $e''$ we have $e \hbefore{\exseq} e'' \hbefore{\exseq} e'$), and
  \item
    $\event'$ cannot be enabled or disabled by an event in another thread.
\end{inparaenum}
\qed
\end{definition}

\begin{definition}[Weak Initials]
\label{def:initials}
For an execution $\exseq.w$, the set of \emph{weak initials} of $w$ (after $\exseq$),
denoted $\wfirsttrans{\exseq}{w}$, 
is the set of threads $p$ such that 
$\exseq.w\mtequiv \exseq.p.(w \setminus p)$ if $p$ is in $w$, and
$\exseq.w.p \mtequiv \exseq.p.w$ if $p$ is not in $w$.
\qed
\end{definition}
Intuitively, $p \in \wfirsttrans{\exseq}{w}$ if $\nextof{\exseq}{p}$ is independent with
all events that precede it in $w$ in the case that $p$ is in $w$, otherwise with all events in $w$.
If $p \in \wfirsttrans{\exseq}{w}$ we say that $w$ is \emph{redundant} wrt.\ $E.p$,
since some extension of $E.w$ is equivalent to some extension of $E.p$.
An important property of the execution tree $\exseqs$ that is
maintained by the algorithm is that
an extension $w$ of an existing sequence $\exseq$ is added only if $\exseqs$ does not contain
an execution of form $E'.p$ such that $E'$ but not $E'.p$ is a prefix of $E$,
and $w'.w$ is redundant wrt.\ $E'.p$, 
where $\exseq'$ is defined by $\exseq = \exseq'.w'$.

For the \gbalg algorithm, we define
\begin{itemize}
\item $\pre(\exseq, \event)$ as the prefix of $\exseq$ up to but not including
  $\event$,
\item $\notsucc{\event}{\exseq}$ as the subsequence of $\exseq$ of events that occur
  after $\event$ but do not happen-after $\event$.
\item $u \infirstseqs{\exseq} w$ to denote that $\exseq.u.v \mtequiv \exseq.w$ for some $v$; intuitively $u$ is a ``happens-before prefix'' of $w$.
\end{itemize}

\begin{figure}
  \newcommand{\lIfElse}[3]{\lIf{#1}{#2 \textbf{else}~#3}}
  \SetKwBlock{Forever}{forever do}{end}
  \removelatexerror \begin{algorithm}[H]
  \SetInd{0.3em}{0.5em}
  \Initial{$\exploregb(\emptyseq)$~with~$\wut(\emptyseq)=~\emptytree$,~$\sleepattr(\emptyseq)=~\emptyset$}
  \BlankLine
  \Fn{$\exploregb(\exseq)$}{
  \If{$\enabledafter{\exseq} = \emptyset$ }
  {
    \ForEach{$\event,\event' \in \domof{\exseq} \ \keyword{such that} \ (\event
      \mayreverserace{\exseq} \event')$}{
       \label{algl:optimal-blocking2-race-begin}
       \label{algl:optimal-blocking2-nbrace-begin}
      $\keyword{let} \ \exseq' = \pre(\exseq,\event)$\; \label{algl:optimal-blocking2-nbrace-Eprime}
\keyword{let} \mbox{$v = (\notsucc{\event}{\exseq}.\procof{\event'})$}\;  \label{algl:optimal-blocking2-nbrace-v}
      \lIf{$\sleepattr(\exseq') \cap \wfirsttrans{\exseq'}{v} = \emptyset$}{\label{algl:optimal-blocking2-nbrace-if}$\insertseq{\exseq'}{v}{\wut(\exseq')}$
      }\label{algl:optimal-blocking2-nbrace-insert}
    } \label{algl:optimal-blocking2-nbrace-end}
    \ForEach{$\tuple{e', \exseq'} \in (\{\langle\nextof{E}{p}, \exseq \rangle | \! $ {\rm $p$ is blocked after $\exseq$}$\} $  \label{algl:optimal-blocking2-brace-begin} \\
      $\qquad   \quad \cup\; \{\tuple{\event', \pre(\exseq,\event')} | $ {\rm $e'$ is in $E$ and may block}$\})$}{ \label{algl:optimal-blocking2-test-disables}
$\canstop := \false$\;
      \ForEach{\rm $\event$ in $\exseq'$ (starting from the end)\\
        $\qquad\qquad$ that may enable or disable $e'$ \label{algl:optimal-blocking2-foreach-e} }{
      $\keyword{let} \ \exseq'' = \pre(\exseq,\event)$\;
        $\keyword{let} \ w = \notsucc{\event}{\exseq}$\;
        \lIf{$\event$ {\rm conflicts with all events that may\\
            $\quad\;\;$ enable or disable} $e'$}{$\canstop := \true$
        }
        $\didinsert := \false$\;
        \ForEach{{\rm maximal subsequence} $u$ of $w$ \keyword{such that}\\
          $\quad\;$ $u \infirstseqs{\exseq''} w$ \keyword{and} {\rm $e'$ is enabled after $\exseq''.u$}}{ \label{algl:optimal-blocking2-subsequence-u}
$\didinsert := \true$\;
          $\keyword{let} \ v = u.\procof{e'}$\;\label{algl:optimal-blocking2-brace-v}
          \lIf{$\sleepattr(\exseq'') \cap \wfirsttrans{\exseq''}{v} = \emptyset$ \label{algl:optimal-blocking2-brace-if}}{$\insertseq{\exseq''}{v}{\wut(\exseq'')}$
          }\label{algl:optimal-blocking2-brace-insert}
        }
        \lIf{$\canstop$ \keyword{and} $\didinsert$}{\keyword{break}}
      }
    } \label{algl:optimal-blocking2-race-end}
  }
  \Else{
\If{$\wut(\exseq) = \emptytree$ }{\label{algl:optimal-blocking2-explore-begin}$\keyword{choose}\; p \in \enabledafter{\exseq}$\;
      $\wut(\exseq) := \tuple{\set{p},\emptyset}$\;\label{algl:optimal-blocking2-wutinit-fresh}
    }  \label{algl:optimal-blocking2-wutinit-end}
\While{$\exists p \in \wut(\exseq)$}{
      $\keyword{let}\; p = \min_{\prec}
               \{ p \in \wut(\exseq) \}$\;\label{algl:optimal-blocking2-leftchild}
      $\sleepattr(\exseq.p) := \lbrace q \in \sleepattr(\exseq) \mid p, q $ {\rm independent after} $E \rbrace $\; \label{algl:optimal-blocking2-sleepset-update}
      $\wut(\exseq.p) := \subtreeafter{p}{\wut(\exseq)}$\; \label{algl:optimal-blocking2-wut-update}
      $\exploregb(\exseq.p)$\;\label{algl:optimal-blocking2-reccall}
      {add $p$ to $\sleepattr(\exseq)$}\;
      {remove all sequences of form $p.w$ from $\wut(\exseq)$}\;
    }  \label{algl:optimal-blocking2-explore-end}
  }
  }
  \Fn {$\insertseq{\exseq'}{v}{\wut(E')}$}{ \label{alg:wut-insert-begin}
    $u := \emptyseq$ \\
\keyword{let} c be the list of children of $u$ in $\wut(E')$ from left to right\;\label{alg:wut-insert-forever}
      \ForEach{\mbox{\rm sequence $u.p$ in $c$}}{
        \label{algl:wut-foreach-child}
        \If{$p \in \wfirsttrans{\exseq'.u}{v}$}{\label{algl:wut-if-in-wi}
          \lIf{$p \not\in v$ \keyword{or} $(v:= v\setminus p) = \emptyseq$}{\Return{}} \label{algl:wut-insert-exhaust}
$u := u.p$\; \label{algl:wut-insert-descend}
\lIf{$u$ \rm{is a leaf of $\wut(E')$}}{\Return{}} \label{algl:wut-insert-restart}
          \keyword{goto} \cref{alg:wut-insert-forever} \label{algl:wut-insert-loop-end}
        }
      }
      add $v$ as a new rightmost descendant of $u$ in $\wut(E')$\label{algl:wut-insert-insert} \\
      \Return{} \label{alg:wut-insert-end}
}
  \caption{\gbalg}
  \label{alg:dpor-alg-blocking2}
  \end{algorithm}
  \vspace{\afteralgorithmspace}
\end{figure}

The algorithm runs in two phases:
race detection
\mbox{(\crefrange{algl:optimal-blocking2-race-begin}{algl:optimal-blocking2-race-end})}
and exploration
(\crefrange{algl:optimal-blocking2-explore-begin}{algl:optimal-blocking2-explore-end}).
Exploration picks the next unexplored leaf of the exploration tree and extends it
with arbitrary scheduling to a maximal execution. This leaf
is reached step-by-step: at each step, the current execution
$\exseq$ is extended by the leftmost child of the root of $\wut(\exseq)$ and used in a
recursive call to $\exploregb$
(\crefrange{algl:optimal-blocking2-leftchild}{algl:optimal-blocking2-reccall}) in order to
perform the next step.
If $\wut(\exseq)$ only contains the empty sequence, an arbitrary thread is chosen for
the next step and added to $\wut(\exseq)$ (\cref{algl:optimal-blocking2-wutinit-fresh}).
This step-by-step extension of the current execution is continued until a maximal execution is
reached. At each step,
the new sleep set after $\exseq.p$ is constructed by taking the elements of
$\sleepattr(E)$ that are independent with $p$.
After a recursive call to $\exseq.p$, the subtree rooted at $\exseq.p$ can be removed from
the wakeup tree. To remember that we should not attempt to explore any
sequences that are redundant wrt.\ $E.p$, we add $p$ to $\sleepattr(E)$.

The race detection phase is entered when the explored sequence $\exseq$ is maximal.
There we examine $\exseq$ for races and construct new non-redundant executions.
We distinguish between two types of races:
non-blocking races, such as between a write and a read, handled on
\crefrange{algl:optimal-blocking2-nbrace-begin}{algl:optimal-blocking2-nbrace-end},
and blocking races, such as involving an \await event, handled on
\crefrange{algl:optimal-blocking2-brace-begin}{algl:optimal-blocking2-race-end}.

For each non-blocking race $\event \mayreverserace{\exseq} \event'$, we let $E'$
be the prefix of $E$ that precedes $e$, and construct a wakeup sequence $v$ by appending
$\procof{e'}$ to the subsequence of events that occur after $e$ in $E$ but do
not happen-after $e$ (\cref{algl:optimal-blocking2-nbrace-v}).
By construction, the sequence $E'.v$ is an execution. Moreover
$\procof{e} \not\in \wfirsttrans{\exseq'}{v}$ since the occurrence of $e'$ in
$v$ does not happen-after $e$. Thus, $v$ is non-redundant wrt.\ $E'.\procof{e}$.
If $v$ is also non-redundant wrt.\ $E'.p$ for each $p \in \sleepattr(E')$,
then $v$ is inserted into the wakeup tree at $E'$,
extending $\wut(E')$ with a new leaf if necessary.

Races involving events that can be blocked are handled
at~\crefrange{algl:optimal-blocking2-brace-begin}{algl:optimal-blocking2-race-end}.
For each such event $e'$, we extract the prefix $E'$
that precedes $e'$. Then, for each $e$ in $E'$ that potentially conflicts with
$e'$, we extract the prefix $E''$ preceding $e$ and the sequence $w$ of events
that does not happen-after $e$. For each maximal happens-before prefix $u$ of
$w$ after which $e'$ is enabled, we construct a wakeup sequence $v$ as
$u.\procof{e'}$ (\cref{algl:optimal-blocking2-brace-v}), which is checked for
redundancy and possibly inserted into the wakeup tree in the same way as for a
nonblocking race.
Such prefixes
can be enumerated by
recursively removing the suffix of one event that may enable or disable $e'$ at
a time, stopping
whenever $e'$ is enabled by the current prefix.
As an optimisation, implemented by the flags $\canstop$ and $\didinsert$,
once the algorithm has found a wakeup sequence that enables $e'$ before some
event that conflicts with every event that may enable or disable $e'$,
it needs not consider reversing $e'$ with even earlier events $e$, as those
reversals will be considered in a later recursive call.

The function $\insertseq{\exseq}{v}{\wut(E')}$ for inserting a sequence $v$ into a wakeup tree
$\wut(E')$ is shown in \crefrange{alg:wut-insert-begin}{alg:wut-insert-end}.
Starting from the root, represented by the empty sequence, it traverses $\wut(E')$ downwards (the current point being $u$), always descending (\cref{algl:wut-insert-descend}) to the leftmost child $u.p$ such that $p$ is a weak initial of the remainder of $v$ until either
\begin{inparaenum}[(i)]
\item arriving at a leaf indicating that $v$ was redundant to begin with and $\wut(E')$ can be left unchanged (\cref{algl:wut-insert-restart}),
\item
  encountering a $p$ which is not in $v$, or exhausting $v$ (\cref{algl:wut-insert-exhaust}), or
\item
  arriving at a node with no child passing the test at~\cref{algl:wut-if-in-wi}, and then adding the
  remainder of $v$ as a new leaf (\cref{algl:wut-insert-insert}), since it was shown to be non-redundant.
\end{inparaenum}

Algorithm \gbalg is correct and optimal in the sense that it explores exactly
one maximal execution in each equivalence class, as stated in the following
theorem
\iffwithappendix{whose proof is in \cref{sect:correctness-final}}{whose proof is
  in the extended version of this paper~\cite{godot-extended@arXiv-22}}.
\begin{theorem}
\label{thm:correct-optimal}
For a terminating program $\prog$, \gbalg has the properties that
\begin{inparaenum}[(i)]
\item for each maximal execution $E$ of $\prog$, it explores some
  execution $E'$ with $E' \mtequiv E$, and
\item it never explores two different but equivalent maximal executions.
\end{inparaenum}
\end{theorem}

\section{Implementation and Evaluation} \label{sec:eval}
We have implemented our techniques on top of the \Nidhugg tool.
\Nidhugg is a state-of-the-art stateless model checker for C/C++ programs
with Pthreads, which works at the level of LLVM Intermediate Representation
(IR), typically produced by the Clang compiler.
We have added our PLP analysis and transformations, as well as the rewrite
from load-assume, exchange-assume, and compare-exchange-assume pairs into
load-await and exchange-await, as passes over LLVM IR.
\Nidhugg comes with a selection of SMC algorithms. One of them is Optimal-DPOR,
which we have used as a basis for our implementation of \algname including IFAA,
the optimisation of treating \mbox{fetch-and-add} instructions to the same
memory location as independent.
All the techniques in this paper are now included in upstream \Nidhugg and are
enabled when giving the \texttt{--optimal} flag.

\subsection{Overall Performance}
First, we evaluate our technique and compare its performance against
baseline \Nidhugg and the \Saver~\cite{SAVER@FMCAD-21} technique, implemented
in a recent version of \GenMC~\cite{GenMC@CAV-21}.
\Saver has a similar goal to our PLP transformation, but tries to identify
pure loop iterations dynamically, aborting threads if they perform a pure loop
iteration. \Saver's approach does not allow further rewrite with \await{}s.

For our evaluation, we used a set of real-world benchmarks similar to those
used by the \Saver~\cite{SAVER@FMCAD-21} paper. We note that all atomic memory
accesses in these benchmarks have been converted to SC, as this is the only
common memory model that both tools support.
Where relevant, benchmarks are ran with the same loop bound as in the \Saver
paper. For most benchmarks, this is one greater than the number of threads.
After the benchmark name, the number of threads are shown in parentheses.
Benchmarks \bench{mcslock}, \bench{qspinlock} and \bench{seqlock} are tests of
data structures from the Linux kernel.
Benchmarks \bench{ttaslock} and \bench{twalock} are mockups based on, but not
the same as, the benchmarks in the \Saver paper, because its authors were not
at liberty to share the original benchmark sources. Both are tests of locking
algorithms.
Benchmark \bench{mpmc-queue} tests a multiproducer-multiconsumer queue
algorithm, \bench{linuxrwlocks} tests a readers-writers lock
algorithm, \bench{treiber-stack} tests a lock-free stack algorithm, and
\bench{ms-queue} tests a lock-free queue.
Benchmarks \bench{mutex} and \bench{mutex-musl} test two mutex algorithms, the
second one used in the musl C standard library implementation.
Benchmark \bench{sortnet} is an extended version of the concurrent sort
program from \cref{fig:prodcons}.
In this version, the sorting networks are generated using
\href{https://en.wikipedia.org/wiki/Batcher_odd\%E2\%80\%93even_mergesort}{Batcher's odd-even mergesort}.
The number of elements sorted is twice the number of threads, so
\bench{sortnet(6)} sorts 12 elements. In our replication
package~\cite{godot-artifact@FMCAD-22}, all the tools and benchmarks are
provided, as well as scripts that can replicate the tables in this section.

\pgfplotstableset{columns/benchmark/.style={string type,
    string replace*={_}{\_},
  },
  columns/tool/.style={string type},
  columns/ref_traces/.style={column name=\multicolumn{1}{c}{Baseline}},
  columns/gref_traces/.style={column name=\multicolumn{1}{c}{Baseline}},
  columns/saver_traces/.style={column name=\multicolumn{1}{c}{\Saver}},
  columns/nplp_traces/.style={column name=\multicolumn{1}{c}{PLP}},
  columns/nplpa_traces/.style={column name=\multicolumn{1}{c}{PLP+Await}},
  columns/nall_traces/.style={column name=\multicolumn{1}{c}{\dots+IFAA}},
  every head row/.style={before row={\toprule
      & \multicolumn{1}{c}{\GenMC} & \multicolumn{4}{c}{\Nidhugg} \\
      \cmidrule(lr){2-2} \cmidrule(lr){3-6}
    },after row=\midrule},
  every last row/.style={after row=\bottomrule},
  column type={r},
create on use/Benchmark/.style={
create col/assign/.code={\getthisrow{benchmark}\benchmark
      \getthisrow{n}\n
      \edef\entry{\benchmark(\n)}\pgfkeyslet{/pgfplots/table/create col/next content}\entry
    },
  },
  columns/Benchmark/.style={
    column type={l},
  },
  columns={Benchmark,
saver_traces,ref_traces,nplp_traces,nplpa_traces,nall_traces},
  fixed,
  string type, }

\newcommand\SZ{\scriptsize}
\newcommand\error{\textcolor{red}{\SZ \textdagger}\xspace}
\newcommand\timeout{\textcolor{blue}{\SZ \clock}\xspace}\newcommand\notavail{\textcolor{gray}{\SZ n/a}\xspace}
\newcommand\bug{\SZ\color{red} bug }

\begin{table*}[ht]
  \caption{Number of (complete+blocked) executions explored by algorithms
    implemented in \GenMC and \Nidhugg on a set of challenging benchmarks, as
    well as the execution time (in seconds) taken.
The \timeout symbol means that the exploration did not finish in 1h, and
    \error means that the tool crashed.}
  \label{tab:results:vs-saver}
  \label{tab:results:vs-saver-full}
\centering
  \setlength{\tabcolsep}{3pt}
\pgfplotstablevertcat{\output}{results/zenlaban/qspinlock.txt}
\pgfplotstablevertcat{\output}{results/zenlaban/mcs_spinlock.txt}
  \pgfplotstablevertcat{\output}{results/zenlaban/twalock.txt}
  \pgfplotstablevertcat{\output}{results/zenlaban/mutex_musl.txt}
  \pgfplotstablevertcat{\output}{results/zenlaban/mutex.txt}
\pgfplotstablevertcat{\output}{results/zenlaban/ms_queue.txt}
  \pgfplotstablevertcat{\output}{results/zenlaban/linuxrwlocks.txt}
\pgfplotstablevertcat{\output}{results/zenlaban/ttaslock.txt}
\pgfplotstablevertcat{\output}{results/zenlaban/seqlock.txt}
  \pgfplotstablevertcat{\output}{results/zenlaban/mpmc_queue.txt}
\pgfplotstablevertcat{\output}{results/zenlaban/treiber_stack.txt}
  \pgfplotstablevertcat{\output}{results/zenlaban/sortnet.txt}
  \resizebox{\textwidth}{!}{
    \pgfplotstabletypeset[
      every row no 2/.style={before row=\midrule},
      every row no 4/.style={before row=\midrule},
      every row no 6/.style={before row=\midrule},
      every row no 8/.style={before row=\midrule},
      every row no 10/.style={before row=\midrule},
      every row no 12/.style={before row=\midrule},
      every row no 14/.style={before row=\midrule},
      every row no 16/.style={before row=\midrule},
      every row no 18/.style={before row=\midrule},
      every row no 20/.style={before row=\midrule},
      every row no 22/.style={before row=\midrule},
      every row no 25/.style={before row=\midrule},
      columns/ref_traces/.style={column name=\multicolumn{1}{c}{Execs}},
      columns/gref_traces/.style={column name=\multicolumn{1}{c}{Execs}},
      columns/saver_traces/.style={column name=\multicolumn{1}{c}{Execs}},
      columns/nplp_traces/.style={column name=\multicolumn{1}{c}{Execs}},
      columns/nplpa_traces/.style={column name=\multicolumn{1}{c}{Execs}},
      columns/nall_traces/.style={column name=\multicolumn{1}{c}{Execs}},
      columns/ref_time/.style={column name=\multicolumn{1}{c}{Time}},
      columns/gref_time/.style={column name=\multicolumn{1}{c}{Time}},
      columns/saver_time/.style={column name=\multicolumn{1}{c}{Time}},
      columns/nplp_time/.style={column name=\multicolumn{1}{c}{Time}},
      columns/nplpa_time/.style={column name=\multicolumn{1}{c}{Time}},
      columns/nall_time/.style={column name=\multicolumn{1}{c}{Time}},
      every head row/.style={before row={\toprule
          & \multicolumn{4}{c}{\GenMC} & \multicolumn{8}{c}{\Nidhugg} \\
          \cmidrule(lr){2-5} \cmidrule(lr){6-13}
          & \multicolumn{2}{c}{Baseline} & \multicolumn{2}{c}{\Saver}
          & \multicolumn{2}{c}{Baseline} & \multicolumn{2}{c}{PLP}
          & \multicolumn{2}{c}{PLP+Await} & \multicolumn{2}{c}{\dots+IFAA} \\
          \cmidrule(lr){2-3}   \cmidrule(lr){4-5}
          \cmidrule(lr){6-7}   \cmidrule(lr){8-9}
          \cmidrule(lr){10-11} \cmidrule(lr){12-13}
        },after row=\midrule},
      columns={Benchmark,gref_traces,gref_time,saver_traces,saver_time,ref_traces,ref_time,nplp_traces,nplp_time,nplpa_traces,nplpa_time,nall_traces,nall_time},
    ]{\output}
  }
  \vspace{\aftertablespace}
\end{table*}

We evaluate all techniques based on the number of executions they explore.
In fact, we show this number using an addition of form $T + B$, where $T$
is the number of explored completed executions and $B$ is the number of
executions that are blocked in the sense that either an \await is deadlocked
or some thread is blocked for executing \stmt{\textbf{assume}(false)} (in
\Nidhugg) or a pure loop iteration (in \Saver). We remark that the \Saver
paper reports only the $T$ part, but, as we will see, often the number of
blocked executions is significant and outnumbers the number of explored
completed executions.  Obviously, both numbers contribute to the time an SMC
tool takes to explore these programs.
The evaluation was performed on a Ryzen 5950X running a July 2022 Arch Linux
system.

In \cref{tab:results:vs-saver},
there are four sets of \Nidhugg columns. Baseline shows the performance of
unmodified \Nidhugg/Optimal. The PLP columns shows the performance of using
unmodified \Nidhugg/Optimal together with Partial Loop Purity Elimination.
Pure loops are bounded with \assume{}s. The PLP+Await columns shows the result
of PLP and transforming \assume{}s into \await{}s, where possible. Finally,
the \dots+IFAA columns report results from when \gbalg treats atomic
fetch-and-add operations as independent.
For the two sets of \GenMC columns, the \Saver columns show the performance
of \GenMC v0.6, which implements the \Saver technique, and the Baseline
columns show the performance of \GenMC v0.5.3, which does not.
The timeout we have used for these benchmarks is 1 hour.

Starting at the top of \cref{tab:results:vs-saver}, \bench{qspinlock} is a
benchmark that does not benefit from \Saver nor PLP, but establishes that the
baseline algorithms of both tools are very similar but \GenMC is faster.
In the next four benchmarks (\bench{mcslock}, \bench{twalock}, \bench{mutex},
and \bench{mutex-musl}), both PLP and \Saver are ineffective, but \await{}s
eliminate most of the blocked traces (in \bench{mcslock}) or all of them (in
the remaining three).  Moreover, we see that IFAA is effective in 
\bench{mutex} and \bench{mutex-musl}, and manages to almost halve the total
number of executions explored.

PLP fails to identify the loop purity in \bench{ms-queue}. The
restriction on the form of purity conditions imposed by our implementation in
\Nidhugg is underapproximating the purity condition to \pcfalse. This
demonstrates a downside with doing purity analysis statically, as \Saver never
needs to represent purity conditions in order to eliminate pure loop
iterations.

In \bench{linuxrwlocks}, PLP is ineffective, because this benchmark does not
contain pure loop iterations as we have defined them. Rather, the loop
contains a pair of fetch-and-add and fetch-and-sub that cancel out, which is
called a ``zero-net-effect'' loop in the \Saver paper~\cite{SAVER@FMCAD-21}.
These are out of scope for a static analysis, as \Saver has to dynamically
undo the elimination if a read appears to have observed the intermediate
effect.
Despite the lack of PLP, \algname significantly speeds up
\bench{linuxrwlocks}.

In \bench{ttaslock}, we believe some implementation issue is preventing \Saver
from eliminating pure loop iterations.
PLP does work, however, and \await{}s eliminate all the blocked executions.

In the next three benchmarks (\bench{seqlock}, \bench{mpmc-queue} and
\bench{treiber-stack}), PLP discovers the same pure loop iterations as \Saver,
and permits a rewrite to \await{}s that significantly reduces the
search space, even by an order of magnitude for \bench{seqlock}, and on
\bench{mpmc-queue} IFAA further halves it.

Finally, \algname really shines on \bench{sortnet}. \GenMC
cannot take advantage of \await{}s, and so has to explore an exponential
number of (assume-blocked) traces, where \Nidhugg can explore the program in
just one.
Unfortunately, \GenMC v0.5.3 crashes on this benchmark, but we believe it
would yield the same numbers as \Saver, which also explores a significant
number of redundant executions.

\begin{table*}
  \caption{Number of (complete+blocked) executions that SMC algorithms
    in \GenMC and \Nidhugg explore on shortened, bug-free versions of
    \bench{\scriptsize safestack}.}
  \label{tab:results:safestack-full}
\centering
  \setlength{\tabcolsep}{2pt}
  \pgfplotstablevertcat{\output}{results/zenlaban/safestack_full.txt}
\resizebox{0.73\textwidth}{!}{
  \pgfplotstabletypeset[
    every head row/.style={before row={\toprule
        & \multicolumn{2}{c}{\GenMC} & \multicolumn{4}{c}{\Nidhugg} \\
        \cmidrule(lr){2-3} \cmidrule(lr){4-7}
      },after row=\midrule},
    columns={Benchmark,gref_traces,saver_traces,ref_traces,nplp_traces,nplpa_traces,nall_traces},
  ]{\output}
  }
  \vspace{\aftertablespace}
\end{table*}

\subsection{Effectiveness on SafeStack}
Next, we evaluate the ability of \algname to expose difficult-to-find
bugs in real-world code bases.
The benchmark we will use is called \bench{safestack}.
It was first \href{https://social.msdn.microsoft.com/Forums/en-US/91c1971c-519f-4ad2-816d-149e6b2fd916/bug-with-a-context-switch-bound-5}{posted to the CHESS forum}, and subsequently included in the
SCTBench~\cite{thomsonDB16} and SVComp benchmark suites.
The original \bench{safestack} code attempts to implement a lock-free stack
but contains an ABA bug which is quite challenging for concurrency testing and
SMC tools to find, in the sense that exposing the bug requires at least five
context switches. The test harness is also quite big, containing three threads
each performing four operations on the stack.
Let us refer to this original harness as \bench{safestack-444} to indicate
that each of its three threads performs four operations (pop, push, pop,
push). We will also use shortened versions of this harness: four versions with
just two threads, and four versions where each of the three threads performs
fewer operations.
The smallest harness that exposes the bug is \bench{safestack-331}.

We first compare the two SMC tools and their algorithms on versions of
\bench{safestack} that do not exhibit the bug and thus require exhaustive
exploration of all traces.
\Cref{tab:results:safestack-full} shows the results.
\begin{inparaenum}[]
\item First, notice that the dynamic technique that \Saver implements is
  completely or mostly ineffective in these programs;
  compare it to the baseline numbers.
\item In contrast, PLP achieves significant reduction of the set of executions
  that \Nidhugg explores.
\item Finally, both the transformation of \assume{}s to \await{}s and the IFAA
  optimisation are applicable and result in further reductions in the number
  of explored executions.
\end{inparaenum}
The number of complete traces is 0 on \bench{safestack-311} since the code
does not allow popping the last element, so all traces end up with one thread
livelocking in pop with the queue containing only one element.
For \cref{tab:results:safestack-full}, the timeout used is 10 hours.

With our next and last experiment, using \bench{safestack-331}, we can evaluate
the tools' abilities to expose the bug. Neither \GenMC, with or without \Saver,
nor baseline \Nidhugg find anything after running for more than $2\,000$ hours!
On the other hand, if we run \Nidhugg with PLP, awaits, and IFAA, it discovers
the bug in just $8$ minutes, after exploring $2+2\,453\,474$ traces. How much
of its search space an SMC tool has to search before it encounters a bug can
be up to ``luck'', so to ensure that this result is not due to luck we ``fix''
the bug by commenting out all the assertions in the benchmark and run \Nidhugg
again. This gives us an upper bound on the size of the search space, i.e., how
much would need to be searched to find the bug in the worst case, and also
provides an indication of how long it might take to verify the program after
fixing the bug. On the fixed \bench{safestack-331}, \Nidhugg terminates in
only $24$ minutes after exploring $5\,772+8\,521\,721$ traces.
This demonstrates how the techniques we presented in this paper substantially
reduce the search space on \bench{safestack}, allowing the bug to be found or
its absence verified by an exhaustive SMC technique. To our knowledge, no
other exhaustive technique has ever been able to discover the bug in
\bench{safestack}.

\section{Related Work}
Since SMC tools assume the analysed program to terminate, they must first
bound unbounded loops. Several tools~\cite{tacas15:tso,NoDe:toplas16,KLSV:popl18,GenMC@PLDI-19} have an automatic loop unroller that
is parameterised by a chosen loop bound.
Several SMC tools, including \Nidhugg~\cite{tacas15:tso},
\RCMC~\cite{KLSV:popl18} and \GenMC~\cite{GenMC@PLDI-19}, transform simple
forms of spinloops, such as the one shown in \cref{fig:spinwait:original}, to \assume statements, but only transform simple polling loops that can be recognised
syntactically. We are not aware of any tool that transforms loops into
\await statements, meaning existing tools are susceptible to
scalability problems for programs like the sorting networks shown in \cref{fig:prodcons}. 
An SMC technique that can diagnose livelocks of spinloops under fair
scheduling is \toolnamefont{VSync}~\cite{VSync@ASPLOS-21}. However, to do so
it enforces fairness, and cannot bound the loop even with an \assume, thus
exploring many more traces than tools which transform
spinloops to \assume{}s.

\Saver~\cite{SAVER@FMCAD-21} also aims to block pure loop iterations by introducing \assume statements.
It identifies pure loop iterations dynamically, instead of by static analysis as
in our approach. \Saver's approach allows to detect a larger class of pure loop iterations, but it does not allow further rewrite with \await{}s. Furthermore, our PLP transformation can block a looping thread at any point in the loop, not just at the back edge.
\Saver also employs several smaller program transformations,
such as loop rotation and merging of bisimilar control flow graph nodes,
that can 
increase the number of loops that may qualify as pure.
These transformations are orthogonal to the detection of pure loop iterations, and could also be used in our framework.

Checking for purity of loop iterations is an idea that has appeared in other
contexts, such as to verify atomicity for concurrent data
structures~\cite{FFQ;purity:journal,Lesani:cav14} and to reduce complexity for
model checking them (e.g.,~\cite{AHHJR:sttt}).

The Optimal-DPOR algorithm implemented in \Nidhugg,
handles mutex locks but not \await statements.
In the journal article of the Optimal-DPOR algorithm~\cite{optimal-dpor-jacm},
principles for handling other blocking statements are presented.
Our \gbalg develops these principles into a practical and efficient algorithm,
which we have also implemented in \Nidhugg.
As future work, the Optimal-DPOR with Observers~\cite{observers} algorithm,
which allows two statements to only conflict in the presence of a third event,
could also be extended (potentially at higher cost) to handle \await{}s.

\section{Concluding Remarks}
We have presented techniques for making SMC with DPOR more effective on loops that
perform pure iterations, including a 
static program analysis technique to detect pure loop executions, a program transformation to block and
also remove them, a weakening of the standard conflict relation, and an optimal DPOR
algorithm which handles the so introduced concepts.
We have implemented the techniques in \Nidhugg, showing 
that they can significantly speed up the analysis of concurrent
programs with pure loops, and also detect concurrency errors.

\section*{Acknowledgements}
This work was partially supported by the Swedish Research Council through
grants \#621-2017-04812 and 2019-05466,
and by the Swedish Foundation for Strategic Research through project aSSIsT.
We thank the anonymous FMCAD reviewers for detailed comments and suggestions
which have improved the presentation aspects of our work.

\bibliographystyle{IEEEtranS}
\begin{footnotesize}
\bibliography{bibdatabase}

\def\Nst#1{$#1^{st}$}\def\Nnd#1{$#1^{nd}$}\def\Nrd#1{$#1^{rd}$}\def\Nth#1{$#1^{th}$}
\begin{thebibliography}{10}
\providecommand{\url}[1]{#1}
\csname url@samestyle\endcsname
\providecommand{\newblock}{\relax}
\providecommand{\bibinfo}[2]{#2}
\providecommand{\BIBentrySTDinterwordspacing}{\spaceskip=0pt\relax}
\providecommand{\BIBentryALTinterwordstretchfactor}{4}
\providecommand{\BIBentryALTinterwordspacing}{\spaceskip=\fontdimen2\font plus
\BIBentryALTinterwordstretchfactor\fontdimen3\font minus
  \fontdimen4\font\relax}
\providecommand{\BIBforeignlanguage}[2]{{%
\expandafter\ifx\csname l@#1\endcsname\relax
\typeout{** WARNING: IEEEtranS.bst: No hyphenation pattern has been}%
\typeout{** loaded for the language `#1'. Using the pattern for}%
\typeout{** the default language instead.}%
\else
\language=\csname l@#1\endcsname
\fi
#2}}
\providecommand{\BIBdecl}{\relax}
\BIBdecl

\bibitem{abdulla2014optimal}
\BIBentryALTinterwordspacing
P.~Abdulla, S.~Aronis, B.~Jonsson, and K.~Sagonas, ``Optimal dynamic partial
  order reduction,'' in \emph{Symposium on Principles of Programming
  Languages}, ser. POPL 2014.\hskip 1em plus 0.5em minus 0.4em\relax New York,
  NY, USA: ACM, 2014, pp. 373--384. [Online]. Available:
  \url{http://doi.acm.org/10.1145/2535838.2535845}
\BIBentrySTDinterwordspacing

\bibitem{tacas15:tso}
\BIBentryALTinterwordspacing
P.~A. Abdulla, S.~Aronis, M.~F. Atig, B.~Jonsson, C.~Leonardsson, and
  K.~Sagonas, ``Stateless model checking for {TSO} and {PSO},'' in \emph{Tools
  and Algorithms for the Construction and Analysis of Systems}, ser. LNCS, vol.
  9035.\hskip 1em plus 0.5em minus 0.4em\relax Berlin, Heidelberg: Springer,
  2015, pp. 353--367. [Online]. Available:
  \url{http://dx.doi.org/10.1007/978-3-662-46681-0_28}
\BIBentrySTDinterwordspacing

\bibitem{optimal-dpor-jacm}
\BIBentryALTinterwordspacing
P.~A. Abdulla, S.~Aronis, B.~Jonsson, and K.~Sagonas, ``Source sets: A
  foundation for optimal dynamic partial order reduction,'' \emph{Journal of
  the ACM}, vol.~64, no.~4, pp. 25:1--25:49, Sep. 2017. [Online]. Available:
  \url{http://doi.acm.org/10.1145/3073408}
\BIBentrySTDinterwordspacing

\bibitem{AHHJR:sttt}
\BIBentryALTinterwordspacing
P.~A. Abdulla, F.~Haziza, L.~Hol{\'{\i}}k, B.~Jonsson, and A.~Rezine, ``An
  integrated specification and verification technique for highly concurrent
  data structures,'' \emph{Int. J. Softw. Tools Technol. Transf.}, vol.~19,
  no.~5, pp. 549--563, 2017. [Online]. Available:
  \url{https://doi.org/10.1007/s10009-016-0415-4}
\BIBentrySTDinterwordspacing

\bibitem{observers}
\BIBentryALTinterwordspacing
S.~Aronis, B.~Jonsson, M.~L{\aa}ng, and K.~Sagonas, ``Optimal dynamic partial
  order reduction with observers,'' in \emph{Tools and Algorithms for the
  Construction and Analysis of Systems - 24th International Conference}, ser.
  LNCS, vol. 10806.\hskip 1em plus 0.5em minus 0.4em\relax Cham: Springer, Apr.
  2018, pp. 229--248. [Online]. Available:
  \url{https://doi.org/10.1007/978-3-319-89963-3\_14}
\BIBentrySTDinterwordspacing

\bibitem{Concuerror:ICST13}
\BIBentryALTinterwordspacing
M.~Christakis, A.~Gotovos, and K.~Sagonas, ``Systematic testing for detecting
  concurrency errors in {Erlang} programs,'' in \emph{Sixth {IEEE}
  International Conference on Software Testing, Verification and Validation},
  ser. ICST 2013.\hskip 1em plus 0.5em minus 0.4em\relax Los Alamitos, CA, USA:
  IEEE, Mar. 2013, pp. 154--163. [Online]. Available:
  \url{https://doi.org/10.1109/ICST.2013.50}
\BIBentrySTDinterwordspacing

\bibitem{FFQ;purity:journal}
\BIBentryALTinterwordspacing
C.~Flanagan, S.~Freund, and S.~Qadeer, ``Exploiting purity for atomicity,''
  \emph{IEEE Trans. Software Eng.}, vol.~31, no.~4, pp. 275--291, Apr. 2005.
  [Online]. Available: \url{https://doi.org/10.1109/TSE.2005.47}
\BIBentrySTDinterwordspacing

\bibitem{FG:dpor}
\BIBentryALTinterwordspacing
C.~Flanagan and P.~Godefroid, ``Dynamic partial-order reduction for model
  checking software,'' in \emph{Principles of Programming Languages,
  (POPL)}.\hskip 1em plus 0.5em minus 0.4em\relax New York, NY, USA: ACM, Jan.
  2005, pp. 110--121. [Online]. Available:
  \url{http://doi.acm.org/10.1145/1040305.1040315}
\BIBentrySTDinterwordspacing

\bibitem{Godefroid:popl97}
\BIBentryALTinterwordspacing
P.~Godefroid, ``Model checking for programming languages using {VeriSoft},'' in
  \emph{Principles of Programming Languages, (POPL)}.\hskip 1em plus 0.5em
  minus 0.4em\relax New York, NY, USA: ACM Press, Jan. 1997, pp. 174--186.
  [Online]. Available: \url{http://doi.acm.org/10.1145/263699.263717}
\BIBentrySTDinterwordspacing

\bibitem{Godefroid:verisoft-journal}
\BIBentryALTinterwordspacing
------, ``Software model checking: The {VeriSoft} approach,'' \emph{Formal
  Methods in System Design}, vol.~26, no.~2, pp. 77--101, Mar. 2005. [Online].
  Available: \url{http://dx.doi.org/10.1007/s10703-005-1489-x}
\BIBentrySTDinterwordspacing

\bibitem{GoHaJa:heartbeat}
\BIBentryALTinterwordspacing
P.~Godefroid, R.~S. Hanmer, and L.~Jagadeesan, ``Model checking without a
  model: An analysis of the heart-beat monitor of a telephone switch using
  {VeriSoft},'' in \emph{Proceedings of the {ACM} {SIGSOFT} International
  Symposium on Software Testing and Analysis}, ser. ISSTA.\hskip 1em plus 0.5em
  minus 0.4em\relax New York, NY, USA: ACM, Mar. 1998, pp. 124--133. [Online].
  Available: \url{https://doi.org/10.1145/271771.271800}
\BIBentrySTDinterwordspacing

\bibitem{GHP:caching-fmsd95}
\BIBentryALTinterwordspacing
P.~Godefroid, G.~J. Holzmann, and D.~Pirottin, ``State-space caching
  revisited,'' \emph{Formal Methods in System Design}, vol.~7, no.~3, pp.
  227--241, 1995. [Online]. Available:
  \url{http://dx.doi.org/10.1007/BF01384077}
\BIBentrySTDinterwordspacing

\bibitem{godot-artifact@FMCAD-22}
\BIBentryALTinterwordspacing
B.~Jonsson, M.~Lång, and K.~Sagonas, ``{Replication Package for Awaiting for
  Godot: Stateless Model Checking that Avoids Executions where Nothing
  Happens},'' Aug. 2022, artifact for the FMCAD 2022 paper with the same title.
  [Online]. Available: \url{https://doi.org/10.5281/zenodo.6979940}
\BIBentrySTDinterwordspacing

\bibitem{KLSV:popl18}
\BIBentryALTinterwordspacing
M.~Kokologiannakis, O.~Lahav, K.~Sagonas, and V.~Vafeiadis, ``Effective
  stateless model checking for {C/C++} concurrency,'' \emph{Proc. ACM on
  Program. Lang.}, vol.~2, no. {POPL}, pp. 17:1--17:32, Jan. 2018. [Online].
  Available: \url{https://doi.org/10.1145/3158105}
\BIBentrySTDinterwordspacing

\bibitem{GenMC@PLDI-19}
\BIBentryALTinterwordspacing
M.~Kokologiannakis, A.~Raad, and V.~Vafeiadis, ``Model checking for weakly
  consistent libraries,'' in \emph{Proceedings of the 40th ACM SIGPLAN
  Conference on Programming Language Design and Implementation}, ser. PLDI
  2019.\hskip 1em plus 0.5em minus 0.4em\relax New York, NY, USA: ACM, Jun.
  2019, pp. 96--110. [Online]. Available:
  \url{https://doi.org/10.1145/3314221.3314609}
\BIBentrySTDinterwordspacing

\bibitem{SAVER@FMCAD-21}
\BIBentryALTinterwordspacing
M.~Kokologiannakis, X.~Ren, and V.~Vafeiadis, ``Dynamic partial order
  reductions for spinloops,'' in \emph{Formal Methods in Computer Aided
  Design}, ser. {FMCAD} 2021.\hskip 1em plus 0.5em minus 0.4em\relax IEEE, Oct.
  2021, pp. 163--172. [Online]. Available:
  \url{https://doi.org/10.34727/2021/isbn.978-3-85448-046-4_25}
\BIBentrySTDinterwordspacing

\bibitem{KoSa:spin17}
\BIBentryALTinterwordspacing
M.~Kokologiannakis and K.~Sagonas, ``Stateless model checking of the {Linux}
  kernel's hierarchical read-copy-update (tree {RCU)},'' in \emph{Proceedings
  of International SPIN Symposium on Model Checking of Software}, ser. SPIN
  2017.\hskip 1em plus 0.5em minus 0.4em\relax New York, NY, USA: ACM, 2017,
  pp. 172--181. [Online]. Available:
  \url{https://doi.org/10.1145/3092282.3092287}
\BIBentrySTDinterwordspacing

\bibitem{GenMC@CAV-21}
\BIBentryALTinterwordspacing
M.~Kokologiannakis and V.~Vafeiadis, ``{GenMC}: {A} model checker for weak
  memory models,'' in \emph{Computer Aided Verification - 33rd International
  Conference, {CAV} 2021, Proceedings, Part {I}}, ser. LNCS, vol. 12759.\hskip
  1em plus 0.5em minus 0.4em\relax Springer, Jul. 2021, pp. 427--440. [Online].
  Available: \url{https://doi.org/10.1007/978-3-030-81685-8\_20}
\BIBentrySTDinterwordspacing

\bibitem{Lesani:cav14}
\BIBentryALTinterwordspacing
M.~Lesani, T.~D. Millstein, and J.~Palsberg, ``Automatic atomicity verification
  for clients of concurrent data structures,'' in \emph{Computer Aided
  Verification, {CAV} 2014}, ser. LNCS, A.~Biere and R.~Bloem, Eds., vol.
  8559.\hskip 1em plus 0.5em minus 0.4em\relax Cham: Springer, Jul. 2014, pp.
  550--567. [Online]. Available:
  \url{https://doi.org/10.1007/978-3-319-08867-9\_37}
\BIBentrySTDinterwordspacing

\bibitem{MQBBNN:chess}
\BIBentryALTinterwordspacing
M.~Musuvathi, S.~Qadeer, T.~Ball, G.~Basler, P.~A. Nainar, and I.~Neamtiu,
  ``Finding and reproducing heisenbugs in concurrent programs,'' in
  \emph{Proceedings of the 8th {USENIX} Symposium on Operating Systems Design
  and Implementation}, ser. OSDI '08.\hskip 1em plus 0.5em minus 0.4em\relax
  Berkeley, CA, USA: USENIX Association, Dec. 2008, pp. 267--280. [Online].
  Available: \url{http://dl.acm.org/citation.cfm?id=1855741.1855760}
\BIBentrySTDinterwordspacing

\bibitem{NoDe:toplas16}
\BIBentryALTinterwordspacing
B.~Norris and B.~Demsky, ``A practical approach for model checking {C/C++11}
  code,'' \emph{{ACM} Trans. Program. Lang. Syst.}, vol.~38, no.~3, pp.
  10:1--10:51, May 2016. [Online]. Available:
  \url{http://doi.acm.org/10.1145/2806886}
\BIBentrySTDinterwordspacing

\bibitem{VSync@ASPLOS-21}
\BIBentryALTinterwordspacing
J.~Oberhauser, R.~L. d.~L. Chehab, D.~Behrens, M.~Fu, A.~Paolillo,
  L.~Oberhauser, K.~Bhat, Y.~Wen, H.~Chen, J.~Kim, and V.~Vafeiadis, ``Vsync:
  Push-button verification and optimization for synchronization primitives on
  weak memory models,'' in \emph{Proceedings of the 26th ACM International
  Conference on Architectural Support for Programming Languages and Operating
  Systems}, ser. ASPLOS 2021.\hskip 1em plus 0.5em minus 0.4em\relax New York,
  NY, USA: ACM, 2021, p. 530–545. [Online]. Available:
  \url{https://doi.org/10.1145/3445814.3446748}
\BIBentrySTDinterwordspacing

\bibitem{thomsonDB16}
\BIBentryALTinterwordspacing
P.~Thomson, A.~F. Donaldson, and A.~Betts, ``Concurrency testing using
  controlled schedulers: An empirical study,'' \emph{ACM Trans. Parallel
  Comput.}, vol.~2, no.~4, pp. 23:1--23:37, 2016. [Online]. Available:
  \url{http://doi.acm.org/10.1145/2858651}
\BIBentrySTDinterwordspacing

\bibitem{DBLP:conf/pldi/ZhangKW15}
\BIBentryALTinterwordspacing
N.~Zhang, M.~Kusano, and C.~Wang, ``Dynamic partial order reduction for relaxed
  memory models,'' in \emph{Programming Language Design and Implementation
  (PLDI)}.\hskip 1em plus 0.5em minus 0.4em\relax New York, NY, USA: {ACM},
  Jun. 2015, pp. 250--259. [Online]. Available:
  \url{http://doi.acm.org/10.1145/2737924.2737956}
\BIBentrySTDinterwordspacing

\end{thebibliography}
\end{footnotesize}

\ifwithappendix{
\appendix

\subsection{Proof of \cref{thm:plp-safety}}
\label{app:plp-proof}
We prove \cref{thm:plp-safety}, i.e., correctness and completeness of
partial loop purity elimination, restated as
\cref{thm:ax-plp-safety,thm:ax-plp-optimality}.

\begin{theorem}[Local State Preservation]
  \label{thm:ax-plp-safety}
  Let $\prog'$ be the program resulting from applying Partial Loop Purity
  Elimination to $\prog$. Then each local state $\thstate$ of a thread $p$
  which is reachable in $\prog$ is also reachable in $\prog'$, provided no
  loop of $p$ is unavoidably pure from $\thstate$.
\end{theorem}

We start by proving a weaker statement, namely that PLP only bounds
unavoidably pure executions.
\begin{lemma}
  \label{lem:ax-plp-safety}

Whenever thread $p$ in state $\thstate$ executes a false \assume statement
  in $\prog'$ inserted by the transformation from $\prog$, then $\thstate$ is
  in an unavoidably pure loop.
\end{lemma}
\begin{proof}
  Assume the counterfactual. State $\thstate$ must be in a loop $\loopvar$. Let
  $\exec$ be the execution of~$\prog'$ that leads to state $\thstate$. The
  assumption is that there is a continuation $w$ of $\exec$ (if we ignore the
  \assume{}s), such that $\exec.w$ either exits the loop without returning to
  the header or completes an impure iteration of~$\loopvar$ after $\thstate$.
Also, there must be some path condition $c_i$ that PLP analysis computed for
  $\loopvar$ in~$\prog$ which was \deftrue at $\thstate$.
  However, by the construction of purity conditions this $c_i$ must include
  all branch conditions for
  returning to the header, so $\exec.w$ cannot have exited without returning
  to the header. Furthermore, $c_i$ would force the execution along a specific
  path through the loop, in which no statement changed the value of a global
  variable, took an internal backedge or a backedge along which the header
  is \impureheader{}. This contradicts the assumption. Thus,
  it must be false, and \cref{lem:ax-plp-safety} true.
\end{proof}

\begin{proof}[Proof of \cref{thm:ax-plp-safety}]
  If $\thstate$ is in a loop, it was not bounded by PLP by
  \cref{lem:ax-plp-safety}, given the assumption that no loop of~$p$ is
  unavoidably pure from $\thstate$. If $\exec$ is an execution of~$\prog$ that
  reaches~$\thstate$, then by removing any complete pure loop executions
  in~$\exec$, we obtain an execution $\exec'$ that also reaches~$\thstate$.
  If any final states of any threads (other than $p$) in~$\exec$ are in
  unavoidably pure loops, we furthermore remove them, yielding another
  sequence $\exec''$. This sequence is an execution that reaches $\thstate$
  because the removed events, by being in unavoidably pure loops, cannot have
  modified any global variables, and so~$p$ can still read the sequence of
  values that is required to reach~$\thstate$.
\end{proof}

For the purposes of proving Property \ref{thm:plp-optimality-prop}
(completeness) of \cref{thm:plp-optimality}, we assume that no
underapproximation is performed when \fpcs are combined using logical
connectives. Additionally, we assume that only two types of simplification of
\fpcs are performed. Either
\begin{inparaenum}
  \item two conjunctions $a \wedge p$ and $a \wedge \neg p$ can be merged to a
    single conjunction $a$, and
  \item two terms that reference the same register(s) and one imply the other
    can be replaced by one of them.
\end{inparaenum}

\begin{theorem}[Pure Loop Elimination]
  \label{thm:ax-plp-optimality}
  Let $\prog'$ be the program resulting from applying PLP to $\prog$. Then no
  execution of $\prog'$ exhibits a completed pure loop iteration of some
  thread.
\end{theorem}

\begin{proof}
Assume the counterfactual. Then there is some
execution $\exec$ of $\prog'$ that exhibits a pure loop iteration.
Let $\event$ be the first event in the pure loop iteration, i.e., the first
statement of the header, and let $\event'$ be the last, i.e., after $\event'$,
$\procof{\event}$ will execute the first statement of the header again.
Following $\exec$ in reverse, we can also follow how \fpcs would have been
propagated by PLP, starting at $\event'$, i.e., a backedge to the header. We
will show that an \assume statement that evaluates to false in~$\exec$ must have
been inserted between $\event$ and $\event'$, thus contradicting the
assumption. We show this by showing that the \fpcs at each point along the
pure loop iteration have the property that if they were the purity condition
of the loop, a false-evaluating \assume would have been inserted by PLP along
the path from $\event$ to $\event'$. We start with $\fpcafter{\event'}$, and
then show that this property is preserved by all the transfer rules, thus
showing that it holds for $\fpcbefore{\event}$, and so for the whole loop.
\newcommand{\assfalse}{\ensuremath{\mbox{assumes-false}}\xspace}
Let us call this property \emph{$\assfalse(\fpcvar)$} for some \fpc $\fpcvar$.

Note that we do not need to worry about the ``took internal backedge''
conjunct in the assume, as pure loop executions cannot take inner loops. Such a
conjunct will always be true for the purposes of this proof.

The \fpc computed before the statement of $\event'$ is given by the backedge
to the header. Since the loop execution is pure, the \fpc will be \pclit{g}
where $g$ is the backedge condition. We know that it holds after $\event'$.
Furthermore, since all registers mentioned in~$g$ reach $\event'$, if
\pclit{g} was the loop purity condition the assume would be inserted somewhere
along the path from $\event$ to $\event'$ in~$\exec$.
Now, we will show that this is preserved by all the transfer rules. First, we
consider the rules for some atomic statement $s$. Load is trivial. Stores
are not possible in a pure loop execution. For atomic adds of form
\stmt{\xvar{} +:= \areg{}}, we add \pclit{\areg{} = 0} to each conjunct. But
\pclit{\areg{} = 0} must evaluate to false at~$s$ since the loop execution is
pure, and otherwise the global variable \xvar{} would have been modified by
$s$, thus contradicting the purity of the loop execution. It must also
evaluate to true at any later insertion point, as no more than one definition
of each register can appear in a pure loop execution. The insertion
location of $\fpcbefore{s}$ might be later than that of $\fpcafter{s}$, but
all other terms in $\fpcafter{s}$ must still evaluate to true if moved later.
The transfer rules for atomic exchange and atomic compare-exchange preserve
the property in the same way as the rule for atomic add, and can be proven
similarly.

Now, if $s'$ is the first statement of a block $B$, and $s$ is the last
statement of a block $A$ immediately preceding it in~$\exec$, we know,
inductively, that $\assfalse(\fpcbefore B)$, and then we can show
$\assfalse(\fpcafter A)$ as follows: Let~$g$ be the condition on the $(A,B)$
edge. It evaluates to true at~$s$ and later. Thus, $\assfalse(\fpcedge AB)$
holds because $\fpcedge AB = g \wedge \fpcbefore B$, similarly as for atomic
statements. The condition $\fpcafter A = \fpcedge AB \vee \fpcvar$ for some
$\fpcvar$
by the transfer rules.
As edge guards are mutually exclusive, we have $g \implies \neg \fpcvar$.
Since $g$ evaluates to true, $\fpcafter A$ must evaluate to false.
If $\fpcafter A$
contains $\fpcedge AB$ as a disjunct, then trivially $\assfalse(\fpcafter A)$.
However, because of the limitations on how \afpc may be simplified,
for each disjunct in $\fpcedge AB$, some disjunct in $\fpcafter A$
must imply it and contain a subset of the registers in it. As removing register references from a conjunction in an \assfalse \fpc
cannot move its insertion location away from the
path followed by~$\exec$, we have
$\assfalse(\fpcafter A)$.

We have now shown inductively that the pure loop execution in~$\exec$ must
contain an \assume with a condition that evaluates to false. This contradicts
the assumption that the pure loop execution in~$\exec$ is complete. Thus, the
assumption is false and \cref{thm:ax-plp-optimality} holds.
\end{proof}

The extension to segmentation-faulting instructions can be proven to satisfy
\cref{thm:ax-plp-optimality} similarly.

\subsection{Proof of \cref{thm:correct-optimal}}
\label{sect:correctness-final}

Let us now prove \cref{thm:correct-optimal}, i.e.,
the correctness and optimality of \gbalg.
We begin with correctness, stated as \cref{thm:correctness-final},
whereafter we go to optimality, stated as \cref{thm:godefroid}.

Throughout, we assume a particular completed
invocation of \gbalg. This invocation consists
of a number of terminated calls to $\exploregb(\exseq)$
for some values of~$\exseq$.
Let $\exseqs$ denote the set of executions
that have been explored in some call to $\exploregb(.)$.
Define the
ordering $\propto$ on $\exseqs$ by letting $\exseq \propto \exseq'$ if
$\exploregb(\exseq)$ returned before $\exploregb(\exseq')$.
Intuitively, if
one were to draw an ordered tree that shows how the exploration has proceeded,
then $\exseqs$ would be the set of nodes in the tree, and $\propto$ would be
the post-order between nodes in that tree.
We use $w, w', \ldots$ to range over sequences, 
$\event, \event', \ldots$ to range over events, as well as:
\begin{itemize}
\item ${\exseq} \enables {w}$ to denote that $\exseq.w$ is an execution,
\item $\exseq' \leq \exseq$ to denote that the sequence $\exseq'$ is
  a prefix of the sequence $\exseq$,
\item $\domofafter{\exseq}{w}$
to denote $\domof{\exseq.w} \setminus \domof{\exseq}$, i.e., the events
  in $\exseq.w$ which are in $w$,
\item  $\event \totorder{\exseq} \event'$ to denote that
  $\event$ occurs before $\event'$ in $\exseq$,
  i.e., $\totorder{\exseq}$ is the total order of events,
\item $w \mtequivafter{\exseq} w'$ to denote $\exseq.w \mtequiv \exseq.w'$, and
\item $\mtclass{\exseq}$ to denote the equivalence class of $\exseq$.
\end{itemize}
For an execution $\exseq.w$ and thread $p$,
\begin{itemize}
  \item let $p \in \firsttrans{\exseq}{w}$ denote that $p \in \wfirsttrans{\exseq}{w}$ and $p \in w$, and
    \item let $\indepafter{\exseq}{p}{w}$ denote that
$p \in \wfirsttrans{\exseq}{w}$ and $p \not\in w$, i.e., that $\nextof{\exseq}{p}$ is independent of
      all events in $w$.
\end{itemize}
For an arbitrary execution $\exseq \in \exseqs$, let
$\finalsleep(\exseq)$ denote the value of $\sleepattr(\exseq)$
  at the point when $\exploregb(\exseq)$ returns.

We begin by two useful lemmas.
The first follows from the involved definitions.
\begin{lemma}
  \label{lem:WI-hbprefix}
  If $p \in \wfirsttrans{\exseq}{w}$ and $u \infirstseqs{\exseq} w$ then
    $p \in \wfirsttrans{\exseq}{u}$.
\end{lemma}

\begin{lemma}
\label{lem:invariants}
During the execution of \gbalg, 
a new leaf $\exseq$ is added to  the execution tree $\exseqs$ only if
  there is no previously added execution of form $E'.p$ with $\exseq = \exseq'.w$, such that $E'$ but not $E'.p$ is a prefix of $E$, and $p \in \wfirsttrans{\exseq'}{w}$.
\end{lemma}
\begin{proof}
The invariant is established by examining the steps of
\cref{alg:dpor-alg-blocking2}.
The only step which inserts a new sequence into~$\exseqs$ is
the insertion of a new leaf in a wakeup tree at \cref{algl:wut-insert-insert}
in the function $\insertseq{\cdot}{\cdot}{\cdot}$.
This step inserts a sequence of form $E'.u.v$ after checking that
$p \not\in \wfirsttrans{\exseq'.u}{v}$ for all existing executions
of form $E'.u.p$. Previous rounds of the loop at \crefrange{alg:wut-insert-forever}{algl:wut-insert-loop-end} checked that
$p \not\in \wfirsttrans{\exseq'.u'}{u''.v}$ for all existing executions
of form $E'.u'.p$ with $u = u'.u''$. Also, the test before insertion at
\cref{algl:optimal-blocking2-nbrace-if} and \cref{algl:optimal-blocking2-brace-if} perform the corresponding check for prefixes of $E'$.
\end{proof}

We can now prove that \cref{alg:dpor-alg-blocking2} is correct in the sense that 
for each maximal execution $\exseq$, it explores an execution
in $\mtclass{\exseq}$.
This is formalised in~\cref{thm:correctness-final} below.
Its proof is by induction over the executions in $\exseqs$,
using the order $\propto$ in which invocations $\exploregb(\exseq)$ return.
The inductive step for an execution $E$ is proven 
by contradiction, by making the assumption that some maximal sequence $\exseq.w$ is unexplored
after the call $\exploregb(\exseq)$ returns. The proof then arrives at a contradition through
a sequence of claims.
First it is shown that the assumption implies 
Claim~\ref{clm:not-explored}, which
states that  $p \not\in \wfirsttrans{\exseq}{w}$ for all
$p \in \finalsleep(\exseq)$.
Thereafter, the sequence of Claims
\ref{clm:valid-eq-q-wq-wR}--\ref{clm:w-infirstseqs-q'} establish that
the algorithm must have explored some sequence which exposes a race,
which by 
Claim~\ref{clm:equivalent-satisfies} causes the algorithm to include
a leaf with properties that contradict the initial assumption in the
inductive step, thereby concluding the proof of the theorem.

\begin{theorem}[Correctness of \gbalg]
\label{thm:correctness-final}
Whenever a call to $\exploregb(\exseq)$ returns
during \cref{alg:dpor-alg-blocking2}, then
for all maximal executions $\exseq.w$,
the algorithm has explored
some execution in~$\mtclass{\exseq.w}$.
\end{theorem}

Since the initial call to the algorithm is
$\exploregb(\emptyseq)$,
\cref{thm:correctness-final} implies that
for all maximal executions $\exseq$  the algorithm explores
some execution in $\mtclass{\exseq}$.

\begin{proof}
By induction on the set of executions $\exseq$ that
are explored during the considered execution, using the ordering $\propto$
(i.e., the order in which the corresponding calls to $\exploregb$ returned).

\medskip
{\sl Base Case:}
This case corresponds to the first sequence $\exseq$ for which the call
$\exploregb(\exseq)$ returns. By the algorithm, $\exseq$ is already
maximal, so the theorem trivially holds.

\medskip
{\sl Inductive Step:}
We prove the inductive step for an arbitrary execution $\exseq$ in $\exseqs$ by contradiction.
So, we make the assumption that there exists a sequence $\exseq$ such that
when the call to $\exploregb(\exseq)$ returns,
there is a maximal sequence $\exseq.w$ such that
the algorithm has not explored
any in $\mtclass{\exseq.w}$. To do this, we employ the following inductive hypothesis:

{\sl Inductive Hypothesis:}
The theorem holds for all execution
sequences $\exseq'$ with $\exseq' \propto \exseq$.

Let us continue the proof of the inductive step.
Let $\sleepset$ and $\WuT$ be the values of
$\sleepattr(\exseq)$ and $\wut(\exseq)$, respectively, when the call to
$\exploregb(\exseq)$ is performed.
Later, just before \Cref{clm:valid-eq-q-wq-wR}, we will impose restrictions on how to
choose $w$ among the ones for which $\exseq.w$ is not explored.
We will show that this leads to a contradiction.

For such $w$ to exist, $\exseq$ cannot be maximal, so $\finalsleep(\exseq)$
contains at least one thread.  For $p \in \finalsleep(\exseq)$, define
\begin{itemize}
\item $\exseq_p'$, such that $\exseq_p' \leq \exseq$,
$\exseq_p'.p \in \exseqs$, and $\exseq_p'.p$ is the last
execution of this form that precedes $\exseq$ (w.r.t.\ $\propto$).
If $\exseq.p \in \exseqs$ then $\exseq_p' = \exseq$, otherwise
if $p \in \finalsleep(\exseq)$ and $\exseq_p'$ is a strict prefix of~$\exseq$.
\item $w_p'$ by $\exseq = \exseq_p'.w_p'$.
\end{itemize}
It follows that $p \in \wfirsttrans{\exseq_p'}{w_p'}$.

\begin{claim}
\label{clm:not-explored}
$\wfirsttrans{\exseq}{w} \cap \finalsleep(\exseq) = \emptyset$.
\end{claim}

\begin{proof}By contradiction.
Assume that there is a $p\in\wfirsttrans{\exseq}{w} \cap \finalsleep(\exseq)$.
Since $w$ is maximal, $p\in\wfirsttrans{\exseq}{w}$ implies $p \in w$ since $p$ is
enabled after $\exseq$ and independent with $w$, therefore enabled throughout $w$.
Therefore from \cref{def:initials} there is a $w''$ such that $\exseq.w
\mtequiv \exseq.p.w'' \mtequiv \exseq_p'.w_p'.p.w'' \mtequiv
\exseq_p'.p.w_p'.w''$. By the
inductive hypothesis applied to $\exseq_p'.p$, the algorithm has explored
some execution in $\mtclass{\exseq_p'.p.w_p'.w''} = \mtclass{\exseq.w}$,
which contradicts the initially made assumption about~$E.w$.
\end{proof}

For $p \in \finalsleep(\exseq)$, define
\begin{itemize}
\item $w_p$, as the longest prefix of $w$ such that $\indepafter{\exseq}{p}{w_p}$,
\item $e_p$, as the first event in $\domofafter{\exseq}{w}$ which is not in
  $w_p$.  Such an event $e_p$ must exist, otherwise $w_p = w$, which implies
  $\indepafter{\exseq}{p}{w}$, which implies $p\in\wfirsttrans{\exseq}{w}$,
  which contradicts Claim~\ref{clm:not-explored}.
\end{itemize}
Also define
\begin{itemize}
\item $q \in \finalsleep(\exseq)$, such that $w_q$ is a longest prefix
  among $w_p$.
  If there are several threads $p \in \finalsleep(\exseq)$
  such that $w_p$ is the same
  longest prefix, then pick $q$ such that
  $\exseq_q'.q$ is minimal (w.r.t.\ $\propto$).
\item $\sleepset'$ as the value of $\sleepattr(\exseq_q'.q)$ when
   the call to $\exploregb(\exseq_q'.q)$ is performed.
\end{itemize}
Without loss of generality, we will assume that among all the possible
$w$ for which the hypothesis in the Inductive Step holds
(i.e., that $\exseq.w$ not explored),
we choose $w$ so that $w_q$ (chosen as described above) is as long as possible.
\begin{claim}
\label{clm:valid-eq-q-wq-wR}
${\exseq_q'}\enables{q.w_q'.w_q}$.
\end{claim}
\begin{proof}
  Since ${\exseq_q'}\enables{q}$ (because $\exseq_q'.q$ was actually explored) and $\indepafter{\exseq_q'}{q}{(w_q'.w_q)}$
  (which follows from $\indepafter{\exseq}{q}{w_q}$ and the fact that $q \in \finalsleep(\exseq)$ implies $\indepafter{\exseq_q'}{q}{w_q'}$), it
follows that \mbox{${\exseq_q'}\enables{q.w_q'.w_q}$}.
\end{proof}

\begin{claim}
\label{clm:maximal-not-initial}
$\wfirsttrans{\exseq_q'.{q}}{w_q'.w_q} \cap \sleepset' = \emptyset$.
\end{claim}
\begin{proof}
  \Cref{clm:valid-eq-q-wq-wR} has shown that ${\exseq_q'}\enables{q.w_q'.w_q}$.
  The proof is then by contradiction: Assume that some thread
  $p$ is in $\wfirsttrans{\exseq_q'.{q}}{w_q'.w_q} \cap \sleepset'$.

  By the construction of $\sleepset'$ (i.e., $\sleepattr(\exseq_q'.q)$) at~\cref{algl:optimal-blocking2-sleepset-update},
  the thread $p$ must be in $\sleepattr(\exseq_q')$ just before the call
  $\exploregb(\exseq_q'.q)$ and
  satisfy $\indepafter{\exseq_q'}{p}{q}$.
This together with $p \in \wfirsttrans{\exseq_q'.{q}}{w_q'.w_q}$ implies that
$p\in \wfirsttrans{\exseq_q'}{q.w_q'.w_q}$, which, using
$\exseq_q'.q.w_q'.w_q \mtequiv \exseq_q'.w_q'.w_q.q$ (which follows from
$\indepafter{\exseq_q'}{q}{(w_q'.w_q)}$),
implies
$p\in \wfirsttrans{\exseq_q'}{w_q'.w_q.q}$, which 
implies $p\in \wfirsttrans{\exseq_q'}{w_q'.w_q}$.
Hence, during exploration of $\exseq_q'.w_q'$, no event in
$w_q'$ removes $p$ from the sleep set. Since $p$ was in
$\sleepattr(\exseq_q')$ just before the call to 
  $\exploregb(\exseq_q'.q)$, we have
$p \not \in w_q'$ and $p$
will end up in $\sleepattr(\exseq_q'.w_q')$ and from there in
$\finalsleep(\exseq_q'.w_q')$, which means $p \in \finalsleep(\exseq)$.
It then follows that $p \not \in w_q$, since otherwise we would have
that $\nextof{\exseq}{p}$ would not conflict with any event preceding it in
$w_q$, hence also in $w$, contradicting  $p \in \wfirsttrans{\exseq}{w}$,
thereby violating \Cref{clm:not-explored}.

Therefore $p \not \in w_q'.w_q$ which by $p\in \wfirsttrans{\exseq_q'}{w_q'.w_q}$ entails
$\indepafter{\exseq_q'}{p}{w_q'.w_q}$.  By choice of $q$, we then have
necessarily that $e_p = e_q$ (otherwise $w_p$ would be longer than $w_q$). But
since among the threads $p$ with $e_p = e_q$ we chose $q$ to be the first one
for which a call of the form $\exploregb(\exseq_q'.p)$ was performed,
we have that $p \not\in \sleepattr(\exseq_q')$ just before the call to
$\exploregb(\exseq_q'.q,\sleepset',\cdot)$, whence $p \not\in \sleepset'$.
Thus, we have a contradiction.
\end{proof}

\begin{claim}
\label{clm:q-witness}
Let $z'$ be any sequence such that $\exseq_q'.q.w_q'.w_q.z'$ is maximal
(such a $z'$ can always be found, since
$\exseq_q'.q.w_q'.w_q$ is an execution).
Then, the algorithm explores some sequence
$\exseq_q'.q.z$ in $\mtclass{\exseq_q'.q.w_q'.w_q.z'}$.
\end{claim}
\begin{proof}
From \Cref{clm:maximal-not-initial}, it follows that
$\wfirsttrans{\exseq_q'.{q}}{w_q'.w_q.z'} \cap \sleepset' = \emptyset$.
Therefore, no execution in
$\mtclass{\exseq_q'.q.w_q'.w_q.z'}$ was
explored before the call to $\exploregb(\exseq_q'.q)$,
otherwise, there would be a call $\exploregb(\exseq''.p)$ with
$\exseq''$ a prefix of $\exseq_q'$ and $p \in \sleepset'$, and defining
$w''$ by $\exseq''.w'' = \exseq_q'$, we would have
$\indepafter{\exseq''}{p}{w''}$ and
$p \in\wfirsttrans{\exseq_q'.{q}}{w_q'.w_q.z'}$, thus contradicting
$\wfirsttrans{\exseq_q'.{q}}{w_q'.w_q.z'} \cap \sleepset' = \emptyset$.
By the inductive hypothesis for $\exseq_q'.q$ applied to $w_q'.w_q.z'$,
the algorithm then explores
some sequence $\exseq_q'.q.z$ in $\mtclass{\exseq_q'.q.w_q'.w_q.z'}$.
\end{proof}

By the construction of $w_q$, the event
$\nextof{\exseq_q'}{q}$ conflicts with $e_q$. We have two cases
\begin{enumerate}
\item If $\nextof{\exseq_q'}{q}$ cannot disable $e_q$, then, letting
  $z'$ be $e_q.z''$ in \Cref{clm:q-witness}, we have
$\nextof{\exseq_q'}{q} \mayreverserace{\exseq_q'.{q}.w_q'.w_q.e_q.z''} e_q$.
From $\exseq_q'.q.z \mtequiv \exseq_q'.{q}.w_q'.w_q.e_q.z''$, it follows that
the same race between $\nextof{\exseq}{q}$ and $e_q$ will also occur in
$\exseq_q'.q.z$, that is, we have
$\nextof{\exseq_q'}{q} \mayreverserace{\exseq_q'.{q}.z} e_q$.
Since the sequence $\exseq_q'.q.z$ is actually explored by the
algorithm, it will encounter the race
$\nextof{\exseq_q'}{q} \mayreverserace{\exseq_q'.{q}.z} e_q$
at \cref{algl:optimal-blocking2-race-begin}.
When handling it,
\begin{itemize}
\item
  $\exseq$ in the algorithm will correspond to $\exseq_q'.{q}.z$ in this proof,
\item
  $\event$ in the algorithm will correspond to $\nextof{\exseq_q'}{q}$ in
  this proof,
\item
  $\event'$ in the algorithm will correspond to $e_q$ in this proof, and
\item
  $v = (\notsucc{\nextof{\exseq_q'}{q}}{\exseq_q'.{q}.z} . \procof{e_q})$
  will be the sequence $v$ at \cref{algl:optimal-blocking2-nbrace-v}.
\end{itemize}
\item
If $e_q$ can be blocked by $\nextof{\exseq_q'}{q}$, then
$e_q$ is possibly blocked after $\exseq_q'.{q}.w_q'.w_q$ (but not after $\exseq_q'.w_q'.w_q$).
Let~$w_q''$ be a shortest sequence such that $e_q$ is enabled after $\exseq_q'.{q}.w_q'.w_q.w_q''$, if such a sequence exists. If $e_q$ is not blocked after
$\exseq_q'.{q}.w_q'.w_q$, then $w_q''$ will be empty.
(The case where no such $w_q''$ exists will be considered in the next paragraph.)
Then $\tuple{e_q, \pre(\exseq,e_q)}$ is one of the tuples $\tuple{e', \exseq'}$ constructed at \cref{algl:optimal-blocking2-test-disables}.
We note that $\nextof{\exseq_q'}{q}$ is in~$\exseq'$ and that there is no event in
$\notsucc{\event}{\exseq}$ that conflicts with all events that may enable or disable~$e_q$,
whence $\nextof{\exseq_q'}{q}$ is one possible choice for $e$ at
\cref{algl:optimal-blocking2-foreach-e}.
Let $u$ be a sequence constructed at \cref{algl:optimal-blocking2-subsequence-u} with
$w_q'.w_q \infirstseqs{\exseq_q} u  \infirstseqs{\exseq_q} \notsucc{\nextof{\exseq_q'}{q}}{\exseq}$ after which
$e_q$ is enabled; such a sequence must exist since $w_q'.w_q$ itself is a possible choice.
It now follows that $u \mtequivafter{\exseq_q} w_q'.w_q$, since otherwise the sequence $u.e_q$ would be
a sequence with $\indepafter{\exseq_q}{q}{u}$, which by construction is not explored
after $\exseq$ such that $u$ is longer than $w_q'.w_q$, thereby contradicting the choice of $w$ introduced
just before \Cref{clm:valid-eq-q-wq-wR}.
The sequence $v$ will therefore have the same construction as in case 1.
\par
In the case where no such $w_q''$ exists, the event $e_q$ will be blocked after
any maximal extension $\exseq_q'.{q}.w_q'.w_q.w_q''$ of $\exseq_q'.{q}.w_q'.w_q$.
We can then proceed as in the preceding paragraph.
\end{enumerate}

\begin{claim}
\label{clm:w-infirstseqs-q'}
$w_q'.w_q.\procof{e_g} \infirstseqs{\exseq_q'} v$.
\end{claim}
\begin{proof}
Using the same argument as in case 2) in the preceding paragraph, it can be established that 
$u \mtequivafter{\exseq_q} w_q'.w_q$, due to the choice of $w$ introduced
just before \Cref{clm:valid-eq-q-wq-wR}.
The claim then follows from the construction of $v$.
\end{proof}

Let $w_R$ denote $w_q.\procof{e_q}$.

\begin{claim}
\label{clm:equivalent-satisfies}
$\sleepattr(\exseq_q') \cap \wfirsttrans{\exseq_q'}{w_q'.w_R} = \emptyset$.
\end{claim}

\begin{proof}
  Assume that some thread $p$ is in $\wfirsttrans{\exseq_q'}{w_q'.w_R}$.
  Let us consider two cases.
\begin{enumerate}
  \item If $p \in w_q'$, then it has no event happening before it in $w_q'$, which
    implies that it cannot have been in $\sleepattr(\exseq_q')$ since then it
    could not have been taken out of the sleep set to be executed in $w_q'$.
    Thus $p \not \in \sleepattr(\exseq_q')$.
  \item If $p \not \in w_q'$, then by $p \in \wfirsttrans{\exseq_q'}{w_q'.w_R}$ we have that
    \mbox{$\indepafter{\exseq_q'}{p}{w_q'}$}, which assuming $p \in
    \sleepattr(\exseq_q')$, means that $p$ will still be in the sleep set after
    $E_q'.w_q'$ and therefore $p \in \finalsleep(\exseq)$. Then:
    \begin{enumerate}
    \item If $p \in w_R$, then $p \in \firsttrans{E}{w_R}$ from which we get $p
      \in \firsttrans{E}{w}$ therefore $p \in \wfirsttrans{E}{w}$. Since $p \in
      \finalsleep(\exseq)$ this contradicts Claim~\ref{clm:not-explored}.
    \item If $p \not \in w_R$, then from $p \in \wfirsttrans{\exseq_q'}{w_q'.w_R}$
      and $p \not \in w_q'.w_R$ we have that $\indepafter{\exseq_q'}{p}{w_q'.w_R}$, which implies
      $\indepafter{\exseq_q'.w_q'}{p}{w_R}$, which is equivalent to
      $\indepafter{\exseq}{p}{w_R}$. But then $w_R = w_q.\procof{e_q}$ is a
      prefix of $w_p$, implying that $w_p$ is strictly longer than $w_q$. This
      contradicts the fact that $q$ was chosen as the thread in
      $\finalsleep(\exseq)$ with the longest prefix
      $w_q$ satisfying $\indepafter{\exseq}{q}{w_q}$.
    \end{enumerate}
\end{enumerate}
Therefore, there can be no such $p \in \wfirsttrans{\exseq_q'}{w_q'.w_R}$ and
Claim~\ref{clm:equivalent-satisfies} is proven.
\end{proof}
From \Cref{clm:w-infirstseqs-q',clm:equivalent-satisfies}
and \cref{lem:WI-hbprefix},
we get $\sleepattr(\exseq_q') \cap
\wfirsttrans{\exseq_q'}{v} = \emptyset$.
Thus, the test at~\cref{algl:optimal-blocking2-brace-if} will succeed,
and the sequence $v$ will be inserted into the 
wakeup tree $\wut(\exseq_q')$ (\cref{algl:optimal-blocking2-brace-insert}) by
the  function $\insertseq{\exseq'}{v}{\wut(E')}$ at
\crefrange{alg:wut-insert-begin}{alg:wut-insert-end}.
We first claim that during the insertion, the sequence $u$ will always
satisfy $E_q.u \leq E$ and $v$ will satisfy $u'.w_q.\procof{e_q} \infirstseqs{E_q.u} v$,
where $u.u' = w_q'$.
This is trivially true initially. To see that
it is preserved by each round of the insertion  starting at \cref{alg:wut-insert-forever},
we consider the possible children of form $u.p$. Let $r$ be the thread such that
$E_q'.u.r \leq E$ (if still $E_q'.u < E$).
We know that $E_q'.u.r$ is in $\exseqs$ when $\exploregb(\exseq)$ returns.
Furthermore, for each branch $u.p$ with $E_q.u.p \propto E_q.u.r$, we have that
$p \not \in \wfirsttrans{\exseq.u}{u'.w_q.\procof{e_q}}$ by the Inductive
Hypothesis and the assumption that $\exseq.w$ has not been explored.
On the other hand,
$r \in \wfirsttrans{\exseq.u}{u'.w_q.\procof{e_q}}$, implying that either
$u.r$ is already in $\wut(\exseq_q')$ during the insertion, in which case the
loop will move to the next iteration with invariants preserved, or
$u.r$ is not already in $\wut(\exseq_q')$,
in which case it must be added during the
current insertion and produce a branch $u.v$ such that
$u'.w_q.\procof{e_q} \infirstseqs{E_q.u} v$.
Thus, when $\insertseq{\exseq'}{v}{\wut(E')}$ returns, the exploration tree will contain
an execution of form $E.v$ with $w_q.\procof{e_q} \infirstseqs{E} v$, thereby
contradicting the assumption that $w_q$ is the longest extension of $E$ that has been explored.
This concludes the proof of the inductive step,
and~\cref{thm:correctness-final} is proven.
\end{proof}

Finally, we also prove that \gbalg is optimal in the sense that it
never explores two different but equivalent executions
and never encounters sleep set blocking.
The following theorem establishes that sleep sets alone are sufficient to
prevent exploration of two equivalent \emph{maximal} executions.
It is essentially the same property that
Optimal-DPOR~\cite{optimal-dpor-jacm} guarantees, and originally appeared
as Theorem~3.2 in the paper of Godefroid et al.~\cite{GHP:caching-fmsd95}.

\begin{theorem}
\label{thm:godefroid}
\gbalg never explores two maximal executions which are equivalent.
\end{theorem}
\begin{proof}
Assume that $\exseq_1$ and $\exseq_2$ are two equivalent maximal execution
sequences that are explored by the algorithm. Then they are both in $\exseqs$.
Assume, without loss of generality, that $\exseq_1 \propto \exseq_2$.
Let $\exseq$ be their longest common prefix, and
let $\exseq_1 = \exseq.p.v_1$ and $\exseq_2 = \exseq.v_2$.
By \cref{lem:invariants} and the definition of
$\firsttrans{\exseq}{v_2}$, we have
$p \not\in\firsttrans{\exseq}{v_2}$, which contradicts
$\exseq_1 \mtequiv \exseq_2$ and the maximality of
$\exseq_1$ and $\exseq_2$.
\end{proof}

}

\end{document}